\documentclass[10pt]{article}

\usepackage[margin=1in]{geometry}
\usepackage{amsthm}
\usepackage{amsmath}
\usepackage{graphicx}
\usepackage{comment}
\usepackage{subfigure}
\usepackage{xfrac}
\usepackage[usenames, dvipsnames]{color}
\usepackage{color}
\usepackage{multicol}

\usepackage{ctable}
\usepackage{longtable}
\usepackage{lipsum}
\usepackage{float}
\usepackage{moreverb}
\usepackage[utf8]{inputenc}
\usepackage[english]{babel}
\usepackage{natbib}
\usepackage{multirow}
\usepackage{dcolumn}
\usepackage{siunitx}
\usepackage{dcolumn,booktabs}
\usepackage{array}
\usepackage{caption} 
\usepackage{booktabs}
\usepackage{amsfonts}
\usepackage{adjustbox}
\usepackage{inputenc}
\usepackage{diagbox}
\usepackage{rotating}
\usepackage{pdflscape}
\usepackage{array}
\usepackage[normalem]{ulem}
\usepackage{bm}
\usepackage{bbm}
\usepackage{authblk}
\usepackage{hyperref}
\usepackage{natbib}
\allowdisplaybreaks

\newtheorem{theorem}{Theorem}[section]

\newtheorem{assumption}{Assumption}[section]

\begin{document}

\title{Nonparametric Method for Clustered Data in Pre-Post Factorial Design}

\author[1]{Solomon W. Harrar}
\author[2,*]{Yue Cui}
\affil[1]{Department of Statistics, University of Kentucky, Lexington, Kentucky, United States}
\affil[2]{Department of Mathematics, Missouri State University, Springfield, Missouri, United States.}
\affil[*]{Corresponding author: YueCui, YueCui@MissouriState.edu}
\renewcommand\Affilfont{\itshape\small}

\date{}

\maketitle

\begin{abstract}
	In repeated measures factorial designs involving clustered units, parametric methods such as linear mixed effects models are used to handle within subject correlations. However, assumptions of these parametric models such as continuity and normality are usually hard to come by in many cases. The homoscedasticity assumption is rather hard to verify in practice. Furthermore, these assumptions may not even be realistic when data are measured in a non-metric scale as commonly happens, for example, in Quality of Life outcomes. In this article, nonparametric effect-size measures for clustered data in factorial designs with pre-post measurements will be introduced. The effect-size measures provide intuitively-interpretable and informative probabilistic comparisons of treatment and time effects. The dependence among observations within a cluster can be arbitrary across treatment groups. The effect-size estimators along with their asymptotic properties for computing confidence intervals and performing hypothesis tests will be discussed. ANOVA-type statistics with $\chi^2$ approximation that retain some of the optimal asymptotic behaviors in small samples are investigated. Within each treatment group, we allow some clusters to involve observations measured on both pre and post intervention periods (referred to as complete clusters), while others to contain observations from either pre or post intervention period only (referred to as incomplete clusters). Our methods are shown to be, particularly effective in the presence of multiple forms of clustering. The developed nonparametric methods are illustrated with data from a three-arm Randomized Trial of Indoor Wood Smoke reduction. The study considered two active treatments to improve asthma symptoms of kids living in homes that use wood stove for heating.
\end{abstract}

\textbf{keywords}: {Dependent replicates, Factorial design, Repeated measures, Missing data, Nonparametric relative effect}

\section{Introduction}  

Factorial designs with pre-post repeated measures are frequently used layouts in experimental science. The data from these studies are typically analyzed by means of parametric procedures, e.g. linear mixed effects model. However, the corresponding parametric assumptions are typically hard to meet in practice. Further, the classic parametric models become more inappropriate if non-metric data such as binary or ordinal data are observed, in the sense that group means are meaningless in describing effects. In more complicated designs where clusters are used as sampling or experimental units in repeated measures factorial designs, the problem becomes more challenging. In this setup, the clusters may be assumed to be independent within factor level combinations but the observations on the sub-units or members of the same clusters cannot, in general, be independent. 

Recently, a number of researches have been done to generalize the classical nonparametric methods to  factorial designs. \cite{akritas1994fully} extended the hypothesis of no treatment effect $H_0^F:F_1=F_2$ from the two sample problem to factorial designs by using an appropriate contrast matrix $\textbf{C}$ and expressing the hypothesis of interest as $H_0^F(\textbf{C}):\textbf{CF}=\bm{0}$, where $\textbf{F}=(F_1,\cdots,F_d)^\top$ denotes the vector of distribution functions. However, it is well-known that hypotheses that are constructed based on distribution functions are restrictive in the sense that data sets with heteroscedasticity are not included. Also, obviously that hypotheses formulated in terms of $\textbf{F}$ have very limited interpretability and do not provide informative effect size measures. All these factors motivated researchers to find alternative ways of constructing and testing hypotheses, which do not only maintain favorable properties of nonparametric tests like free from distribution assumptions, but also offer intuitively meaningful interpretation of treatment and time effects. Generalization of the Wilcoxon-Mann-Whitney (WMW) effect \citep{WMW-1947} 
$$w=P(X<Y)+\frac{1}{2}P(X=Y)$$
to factorial designs allowed us to make significant stride in this direction. This WMW effect $w$ quantifies the tendency of $X$ to be smaller than $Y$ if $p > 1/2$ and $X$ to be larger than $Y$ if $p < 1/2$.
On the other hand, X does not tend to be larger nor smaller than Y if $p = 1/2$. \cite{brunner2017rank} extended the hypotheses of no treatment effect $H_0^w:w=\frac{1}{2}$ in the one-sample problem to factorial designs as $H_0^p:\textbf{Cp}=\bm{0}$, where $\textbf{p}=(p_1,\cdots,p_d)^\top$ denotes the vector of the unweighted relative effects $p_i=\int GdF_i$ and $G=d^{-1}\sum_{i=1}^{d}F_i$. There is also an extension of the nonparametric method to  repeated measures in factorial designs in \cite{brunner2017rank}, which allows the subjects to have repeated measurements at multiple time points within each factor level combination. However, subjects are assumed to be independent and observed at all time points.

Nevertheless, none of the aforementioned researches are applicable to factorial designs with clustered data. \cite{cui2020} proposed a nonparametric test for clustered data in pre-post intervention designs, where some clusters are allowed to have subunits in either pre or post intervention periods but not both. This data structure is referred to as partially complete clustered data and can be viewed as arising from missing values. {Without the factorial design structure, \cite{DS-2005} gave an extension of the two group problem to multiple group problem under the clustered data setting, where individuals within each cluster may belong to one of $m$ possible groups. This particular data structure could be treated as longitudinal clustered data with $m$ $(m>2)$ time points. Again, this manuscript formulated hypotheses in terms of distribution functions, using the idea of within-cluster resampling \citep{Hoffman2001}.} While \cite{cui2020} constructed hypothesis on the nonparametric effect size, i.e. $H_0:w=\frac{1}{2}$, the extension to multiple group designs or factorial designs is not straightforward. The correlation between subunits in the same cluster along with the numerous possibilities of complete and incomplete cluster allocations makes the covariance matrix of effect size estimator quite complicated. It is our intention to close this gap and develop asymptotic test and confidence interval procedures that possess the describable properties of nonparametric methods.

The remainder of the paper is organized as follows. Section \ref{sec:MotivatingExample} describes the Asthma Randomized Trial of Indoor Wood Smoke (ARTIS)\citep{ARTIS-data} which provides a motivation for the methods developed in this paper. A precise formulation of the statistical model and hypotheses on relative effects are given in Section \ref{sec:NonParmModelandHypothesis}. Point estimator as well as asymptotic properties of the relative effect size vector $\textbf{p}$ are derived in Section \ref{sec:AsymptoticTheory}. Based on the main results in Section \ref{sec:AsymptoticTheory}, ANOVA-type statistics for the hypothesis $H_0^p:\textbf{Tp}=\bm{0}$ with $\chi^2$ and $F$ based small-sample approximations are developed in Section \ref{sec:TestStat}. Also in Section \ref{sec:TestStat}, an asymptotic confidence intervals for effect sizes are developed. Section \ref{sec:SimuStudy} reports simulation results for the methods proposed in the paper under various realistic settings. The ARTIS data is analyzed by both the proposed nonparametric approach and linear mixed effects model in Section \ref{sec:AnalysisAndComparison}. Discussion and some concluding remarks are provided in Section \ref{sec:DiscussionConclusion}. All technical details are placed in the Appendix.

\section{Motivating Example}
\label{sec:MotivatingExample}
A relevant example is the Asthma Randomized Trial of Indoor wood Smoke (ARTIS) \citep{ARTIS-data,ward2017efficacy,noonan2012asthma}. In this trial, Pediatric Asthma Quality of Life Questionnaire (PAQLQ) was employed to asses the  primary outcome. PAQLQ is a 23-item asthma-specific battery in three domains: symptoms (10 items), activity limitation (5 items), and emotional function (8 items). Response options for each item are on a 7-point scale where 1 indicates maximum impairment and 7 indicates no impairment. In the trial, {98} houses (possibly with multiple children) were randomized to the three interventions: placebo, wood-stove and air-filter. 
In the homes of placebo group, a sham air-filter was installed.  Real air-filter devices were installed in the homes in the air-filter group. The households assigned to the wood-stove group had their wood-stoves upgraded. Data on quality of life were captured using PAQLQ at each of the four visits in  the winters before and after the intervention, which will hereinafter be referred to as pre- and post-intervention periods.  Here, cluster correlation is induced at three levels of clustering: children clustered in the same household, pre- and post-intervention measurements, and  multiple visits in the pre- and post intervention periods.  Therefore, one challenge with ARTIS data is that there are multiple hierarchies of cluster correlation. Another challenge is  the presence of missing values. Responses on some of the PAQLQ items are missing for some  {participants} and  no visit data is available for some of the houses in either the pre- or post-intervention period. 

The objective of the study is to monitor the difference in PAQLQ scores across placebo, wood-stove and air-filter intervention groups, as well as the pattern of changes in scores in pre- and post-intervention periods.  We adopt the terminologies used in  \cite{ARTIS-data} and describe these changes as intervention effect and winter effect (or time effect in the more general case), respectively. Naturally, we also want to see whether the intervention effect interacts with time effect. Note that the outcomes {of interest in this study} are ordered categorical type and, therefore, analyzing them using mean-based or other parametric effect size measures {as in \cite{ward2017efficacy}} would be inappropriate. In the remainder of the paper, we introduce a fully nonparametric approach which not only accommodates non-metric, binary and ordinal data seamlessly, but also provides informative nonparametric effect size measure. The primary interest is in estimating nonparametric intervention and time effects, and testing whether the effects of the interventions, time and their interaction are significant.


\section{Nonparametric Model and Hypothesis}
\label{sec:NonParmModelandHypothesis}
\subsection{Statistical Model}
\label{sec:StatisticalModel}
In this section, we describe the nonparametric model for the general factorial design setup with pre- and post-intervention measurements. Suppose there are $T$ intervention groups and assume there are $n_j$ independent clusters in the $j^{th}$ intervention group. The clusters may be partitioned into two groups containing $n_j^{(c)}$ complete cases and $n_{j1}+n_{j2}$ incomplete cases. A cluster is identified as \textit{complete} if it contains data both before and after the intervention, while a cluster is identified as \textit{incomplete} if it only has data collected either before or after the intervention but not both. This data structure is called partially complete clustered data \citep{cui2020}. 
Here, $n_{j}^{(c)}$ is the number of complete clusters in the $j^{th}$ intervention group and $n_{jl}$ is the number of incomplete clusters in the $j^{th}$ intervention group and during the $l^{th}$ intervention period {(time point)}. Note that $n_j=n_{j}^{(c)}+\sum_{l=1}^{2}n_{jl}$ is the total number of clusters (sample size) in the $j^{th}$ intervention group. 



To represent data for complete or incomplete clusters in a unified way, we distinguish them by adding a superscript identifier. More specifically, the superscript \text{(c)} pertains to complete clusters and \text{(i)} pertains to incomplete clusters. 

Now, denote $X_{jlkv}^{(c)}$ as the response of the $k^{th}$ complete cluster in the $j^{th}$ intervention group during the $l^{th}$ intervention period and at the $v^{th}$ visit. Also let $X_{jlkv}^{(i)}$ to be the response for incomplete clusters defined similarly. The responses for the $k^{th}$ complete cluster in the $j^{th}$ intervention group are assembled into a vector as 
\begin{equation*}
	\label{CompleteModel}
	\bm{X}_{jk}^{(c)}=(\bm{X}_{j1k}^{(c)\top},\bm{X}_{j2k}^{(c)\top})^\top=(X_{j1k1}^{(c)},\cdots,X_{j1km_{j1k}^{(c)}}^{(c)},X_{j2k1}^{(c)},\cdots,X_{j2km_{j2k}^{(c)}}^{(c)})^\top, k=1,\ldots,n_{j}^{(c)},
\end{equation*}
where $m_{jlk}^{(c)}$ is the number of subunits (visits) for the $k^{th}$ complete cluster in the $j^{th}$ intervention group during the $l^{th}$ intervention period. Furthermore, the incomplete data are collected in the vectors
\begin{equation*}
	\label{IncompleteModel}
	\bm{X}_{jlk}^{(i)}=(X_{jlk1}^{(i)},\cdots,X_{jlkm_{jlk}^{(i)}}^{(i)})^\top, l=1,2,k=1,\ldots,n_{jl},
\end{equation*}
where  $m_{jlk}^{(i)}$ is the number of subunits (visits) for the $k^{th}$ incomplete cluster in the $j^{th}$ intervention group during the $l^{th}$ intervention period. For the derivation of asymptotic theories in Section \ref{sec:AsymptoticTheory}, we further denote \[N_j=\sum_{l=1}^{2}N_{jl}=\sum_{l=1}^{2}(\sum_{k=1}^{n_{j}^{(c)}}m_{jlk}^{(c)}+\sum_{k'=1}^{n_{jl}}m_{jlk'}^{(i)})=\sum_{l=1}^{2}(N_{jl}^{(c)}+N_{jl}^{(i)})\]
as the total number of observations in the $j^{th}$ intervention group, where $N_{jl}$ represents the number of observations during the $l^{th}$ intervention period and $N_{jl}^{(A)}$, $A\in\{c,i\}$ represents the sample sizes in complete or incomplete clusters, respectively. In particular, $N=\sum_{j=1}^{T}N_{j}$ represents the total number of observations in the design. The data scheme is displayed in Table \ref{Table:Structure}, where the rows in each intervention group represent data from subunits in the same cluster.
\begin{table}[!htb]
	\footnotesize
	\centering
	\caption{Schematic representation of the data structure.}
	\begin{tabular}{|c|c|c|l|c|c|c|}
		\hline
		& \multicolumn{2}{c|}{Intervention=1} & \multirow{11}{*}{$\cdots$} &  & \multicolumn{2}{c|}{Intervention=T} \\ \cline{1-3} \cline{5-7} 
		{Cluster} & Pre & Post &  & {Cluster} & Pre & Post \\ \cline{1-3} \cline{5-7} 
		1 & $X\cdots X$ & $X\cdots X$ &  & 1 & $X\cdots X$ & $X\cdots X$ \\ \cline{1-3} \cline{5-7} 
		$\vdots$ & $\vdots$ & $\vdots$ &  & $\vdots$ & $\vdots$ & $\vdots$ \\ \cline{1-3} \cline{5-7} 
		$n_{1}^{(c)}$ & $X\cdots X$ & $X\cdots X$ &  & $n_{T}^{(c)}$ & $X\cdots X$ & $X\cdots X$ \\ \cline{1-3} \cline{5-7} 
		$n_{1}^{(c)}+1$ & $X\cdots X$ & MISSING &  & $n_{T}^{(c)}+1$ & $X\cdots X$ & MISSING \\ \cline{1-3} \cline{5-7} 
		$\vdots$ & $\vdots$ & $\vdots$ &  & $\vdots$ & $\vdots$ & $\vdots$ \\ \cline{1-3} \cline{5-7} 
		$n_{1}^{(c)}+n_{11}$ & $X\cdots X$ & MISSING &  & $n_{T}^{(c)}+n_{T1}$ & $X\cdots X$ & MISSING \\ \cline{1-3} \cline{5-7} 
		$n_{1}^{(c)}+n_{11}+1$ & MISSING & $X\cdots X$ &  & $n_{T}^{(c)}+n_{T1}+1$ & MISSING & $X\cdots X$ \\ \cline{1-3} \cline{5-7} 
		$\vdots$ & $\vdots$ & $X\cdots X$ &  & $\vdots$ & $\vdots$ & $\vdots$ \\ \cline{1-3} \cline{5-7} 
		$n_{1}^{(c)}+n_{11}+n_{12}$ & MISSING & $X\cdots X$ &  & $n_{T}^{(c)}+n_{T1}+n_{T2}$ & MISSING & $X\cdots X$ \\ \hline
	\end{tabular}
	\label{Table:Structure}
\end{table}

We further assume that data from different clusters are independent, and
\begin{equation}
	\label{equation:F}
	X_{jlkv}^{(c)},X_{jlkv}^{(i)}\sim F_{jl},
\end{equation}
where $F_{jl}$ is the marginal distribution function for the data in the $j^{th}$ intervention group during the $l^{th}$ intervention period. Note that this assumption is naturally satisfied under the MCAR mechanism. In order to allow for metric, discrete, dichotomous as well as ordinal data in a unified way, we use the normalized distribution function, i.e. $F_{jl}(x)=\frac{1}{2}\{F_{jl}^-(x)+F_{jl}^+(x)\}$, which is the mean of the left-continuous and right-continuous versions of the distribution function. In the context of nonparametric models, the normalized distribution function was first mentioned by \cite{Kruskal-1952} and later used in more general contexts by \cite{Levy-1925}, \cite{Ruymgaart-1980} and \cite{baby-case-2000} for deriving asymptotic properties of rank statistics in the presence of ties.

To describe intervention effects in the nonparametric setting, we shall use the marginal effects of the $j^{th}$ intervention group during the $l^{th}$ intervention period relative to the $r^{th}$ intervention group during the $s^{th}$ intervention period, defined as
\begin{equation}
	\label{equation:wrsjl}
	w_{rsjl}=\int F_{rs}dF_{jl}=P(X_{rs11}^{(A)}<X_{jl11}^{(B)})+\frac{1}{2}P(X_{rs11}^{(A)}=X_{jl11}^{(B)})\quad\text{for}
\end{equation}
\[1\le r,j\le T, 1\le s,l\le2\text{ and }A,B\in\{c,i\}.\]
The nonparametric relative treatment effect of the $j^{th}$ intervention during the $l^{th}$ intervention period with respect to the average distribution is defined by
\begin{equation}
	\label{equation:p_{jl}}
	p_{jl}=\int GdF_{jl},\qquad j=1,\dots,T,\quad l=1,2,
\end{equation}
where $G=(2T)^{-1}\sum_{j=1}^{T}\sum_{l=1}^{2}F_{jl}$ is the unweighted pooled distribution function. Here, it is easy to verify that 
\[p_{jl}=\overline{w}_{\cdot\cdot jl}=\frac{1}{2T}\sum_{r=1}^{T}\sum_{s=1}^{2}w_{rsjl}.\]
In particular, 
$$p_{jl}=P(U_G<X_{jl11})+\frac{1}{2}P(U_G=X_{jl11})$$
for a random variable $U_G\sim G$ which is independent of $X_{jl11}$. It then follows that $p_{jl}>\frac{1}{2}$ means data from distribution $F_{jl}$ tend to be greater than those from the unweighted mean of all the distributions $G$. Assemble the nonparametric effects into a vector as
\begin{equation}
	\label{equation:pvector}
	\textbf{p}=(p_{11},p_{12},\cdots,p_{T1},p_{T2})^\top=\int Gd\textbf{F},
\end{equation}
where $\textbf{F}=(F_{11},F_{12},\cdots,F_{T1},F_{T2})^\top$. For the derivation of theoretical results in the estimation of $\textbf{p}$ in the subsequent sections, we will need some assumptions on the sample and cluster sizes.
\begin{assumption}
	\label{Assusmption:SampleSize}
	$\underset{1\le j\le T}{\textrm{min}}(n_j^{(c)}+n_{jl})\rightarrow\infty$ such that $\frac{n}{n_j^{(c)}+n_{jl}}\le N_0<\infty$ for $j=1,\cdots,T.$
\end{assumption}
\begin{assumption}
	\label{Assumption:ClusterSize}
	$1\le m_{jlk}^{(A)}\le M$ for $j=1,\cdots,T$, $l=1,2$ and $A\in\{c,i\}$.
\end{assumption}
Assumption \ref{Assusmption:SampleSize} is satisfied if the sample sizes for each intervention period is large in all intervention groups. This can happen, for example, if either the complete cases or both of the incomplete and complete cases are large in each group. Assumption \ref{Assumption:ClusterSize} is rather mild and it ensures all cluster sizes are finite. Note, Assumption \ref{Assusmption:SampleSize} together with Assumption \ref{Assumption:ClusterSize} also stimulate conditions on the total number of observations, i.e. $N_{jl}^{(c)}+N_{jl}^{(i)}\rightarrow\infty$ such that $\underset{1\le j\le T}{\textrm{min}}\frac{N}{N_{jl}^{(c)}+N_{jl}^{(i)}}<\infty$ for $j=1,\cdots,T$. 

The setup described above may give the impression that the paper is dealing with one between-subject with levels $T$. However, the index $j=1,\cdots,T$ is to be viewed as lexicographic order of the between-subject factor level. Therefore, the setup covers factorial designs with crossed and nested factors.

\subsection{Decomposition of the Relative Effects}
\label{sec:Decomposition}
According to the definition of relative treatment effect size in (\ref{equation:p_{jl}}), it is interpreted as follows: if $p_{jl}\le p_{rs}$ the observations in the $j^{th}$ intervention group and during the $l^{th}$ intervention period tend to be smaller than those in the $r^{th}$ intervention group and during the $s^{th}$ intervention period. Another interpretation was introduced in \cite{akritas1994fully} and the supporting information in \cite{de2015regression}. Their ideas are to characterize the relative effects as additive effects obtained by a decomposition of the distribution functions, similar to cell mean decomposition in ANOVA. More precisely, define
\[A_j=\overline{F}_{j\cdot}-G=(1/2)\sum_{l=1}^{2}F_{jl}-G,\quad B_l=\overline{F}_{\cdot l}-G=(1/T)\sum_{j=1}^{T}F_{jl}-G,\]
and
\[(AB)_{jl}=F_{jl}-\overline{F}_{j\cdot}-\overline{F}_{\cdot l}+G,\]
such that $$F_{jl}=G+A_j+B_l+(AB)_{jl}.$$ According to (\ref{equation:p_{jl}}), this decomposition yields an additive effects representation similar to the mean decomposition in classical linear model as
\begin{equation*}
	\label{equation:pjlDecomposition}
	p_{jl}=\int GdF_{jl}=\int GdA_j+\int GdB_l+\int Gd(AB)_{jl}+1/2=\alpha_j+\beta_l+(\alpha\beta)_{jl}+1/2,
\end{equation*}
with side conditions $\sum_{j=1}^{T}\alpha_j=\sum_{l=1}^{2}\beta_l=0$, $\sum_{j=1}^{T}(\alpha\beta)_{jl}=0$ for $l=1,2$ and $\sum_{l=1}^{2}(\alpha\beta)_{jl}=0$ for $j=1,\cdots,T$. Here, the additive effect $\alpha_j$ can be expressed as 
\[\alpha_j=\int Gd\overline{F}_{j\cdot}-1/2=P(U_G<U_{Tj})+1/2P(U_G=U_{Tj})-1/2,\]
where $U_G\sim G$ and $U_{Tj}\sim \overline{F}_{j\cdot}$. Similar interpretations hold for $\beta_l$ and $(\alpha\beta)_{jl}$. Therefore, $\alpha_j>0$ is interpreted as a randomly selected observation $U_{jl}$ from the mean distribution $\overline{F}_{j\cdot}$ under the $j^{th}$ intervention group is more likely to be greater than a randomly selected observation $U_G$ from the mean of all distributions $G$.

\subsection{Hypotheses}
\label{sec:HypothesisSetting}
In this section, we introduce hypotheses that can be used to test intervention, time as well as their interaction effects. The hypotheses of interest in the nonparametric factorial design model described in Section \ref{sec:Decomposition} can generally be stated as
\begin{equation}
	\label{equation:NullHypothesis}
	\textbf{H}_0^p:\textbf{Cp}=\textbf{0},
\end{equation}
where hypotheses matrix \textbf{C} is in the same way as the classical linear model. The hypothesis in (\ref{equation:NullHypothesis}) is equivalent to $H_o^p:\textbf{Tp}=\bm{0}$, where
\begin{equation}
	\label{equation:T}
	\textbf{T}=\textbf{C}^\top(\textbf{C}\textbf{C}^\top)^+\textbf{C}
\end{equation} 
is the unique projection matrix in the hypothesis space, see for example \cite{brunner2001nonparametric} and \cite{brunner2017rank}. For example, in the factorial design with pre- and post-intervention measurements, the contrast matrices for the hypothesis of intervention, time and interact effects are
\begin{enumerate}
	\item[(a)] no intervention effect: $H_0^p(A):\textbf{T}_A\textbf{p}=\bm{0}$ where $\textbf{T}_A=\textbf{P}_T\bigotimes\frac{1}{2}\textbf{J}_2$
	\item[(b)] no time effect: $H_0^p(B):\textbf{T}_B\textbf{p}=\bm{0}$ where $\textbf{T}_B=\frac{1}{T}\textbf{J}_T\bigotimes\textbf{P}_2$ and
	\item[(c)] no interaction effect: $H_0^p(AB):\textbf{T}_{AB}\textbf{p}=\bm{0}$ where $\textbf{T}_{AB}=\textbf{P}_T\bigotimes\textbf{P}_2$,
\end{enumerate}
where $\textbf{P}_a=\textbf{I}_a-(1/a)\textbf{J}_a$ denotes the $a$-dimensional centering matrix. With the decomposition in (\ref{equation:pjlDecomposition}), the above hypotheses can be written in terms of the additive effects, e.g. $H_0^p(A):\textbf{T}_A\textbf{p}=\bm{0}$ is equivalent to $H_0(A):\alpha_1=\cdots=\alpha_T=0$. 

Note that the hypotheses for a special case of the $T\times2$ design with independent observations have been discussed by \cite{boos1992rank}. Also, \cite{akritas1997nonparametric} have discussed in detail the meaning and interpretation of general nonparametric effects in terms of distribution functions. The relative nonparametric effect $p_{jl}$ may be viewed as a measure of an overlap of the distribution function $F_{jl}$ with the mean distribution function $G$, and, thus is to be understood in a manner similar to the nonparametric effect defined in terms of the distribution functions.

\section{Asymptotic Theory}
\label{sec:AsymptoticTheory}
\subsection{Estimator of the Effect Size Vector}
To estimate the effect size vector $\textbf{p}$, we first estimate the relative treatment effect size $p_{jl}$ in (\ref{equation:p_{jl}}). In the presence of dependent subunits (visits), \cite{cui2020} estimate the distribution functions by weighted average of empirical distributions of the data from complete and incomplete clusters where the weights are their respective sample sizes. We utilize the same idea to estimate $F_{jl}$ and, thus, $G$ analogously. 

Denote the cluster-level empirical distributions functions as
\[ \widehat{F}_{jlk}^{(c)}(x)=\frac{1}{m_{jlk}^{(c)}}\sum_{v=1}^{m_{jlk}^{(c)}}c(x-X_{jlkv}^{(c)})\quad \textrm{for}\quad  k=1,\cdots,n_{j}^{(c)}\quad\textrm{and}\]
\[ \widehat{F}_{jlk}^{(i)}(x)=\frac{1}{m_{jlk}^{(i)}}\sum_{v=1}^{m_{jlk}^{(i)}}c(x-X_{jlkv}^{(i)})\quad \textrm{for}\quad k=1,\cdots,n_{jl},\]
for complete and incomplete clusters, respectively, where $c(u)$ denotes the normalized count function, i.e. $c(u)=0,\frac{1}{2},1$ if $u<0,u=0$ or $u>0$. More precisely, $ \widehat{F}_{jlk}^{(c)}$ represents the empirical distribution of the $k^{th}$ complete cluster in the $j^{th}$ intervention group during the $l^{th}$ intervention period and $ \widehat{F}_{jlk}^{(i)}$ is defined similarly for the incomplete clusters. Further, denote $ \widehat{F}_{jl}(x)$ as the empirical distribution of the $j^{th}$ intervention group during the $l^{th}$ intervention period, which is defined by weighted average of empirical distributions of the corresponding complete and incomplete clusters, i.e.
$$ \widehat{F}_{jl}(x)=\theta_{jl} \widehat{F}_{jl}^{(c)}(x)+(1-\theta_{jl}) \widehat{F}_{jl}^{(i)}(x),$$ 
where $\theta_{jl}=\frac{N_{jl}^{(c)}}{N_{jl}}$ and $\widehat{F}_{jl}^{(c)}(x)=\sum_{k=1}^{n_{j}^{(c)}}\frac{m_{jlk}^{(c)}}{N_{jl}^{(c)}} \widehat{F}_{jlk}^{(c)}(x)$, $\widehat{F}_{jl}^{(i)}(x)=\sum_{k=1}^{n_{jl}}\frac{m_{jlk}^{(i)}}{N_{jl}^{(i)}} \widehat{F}_{jlk}^{(i)}(x)$ are the weighted average of the empirical distributions of complete and incomplete clusters, respectively, in the $j^{th}$ intervention group during the $l^{th}$ intervention period, respectively. Equivalently, 
\begin{equation}
	\label{equation:empirical}
	\widehat{F}_{jl}(x)=\frac{1}{N_{jl}}\big\{\sum_{k=1}^{n_{j}^{(c)}}\sum_{v=1}^{m_{jlk}^{(c)}}c(x-X_{jlkv}^{(c)})+\sum_{k=1}^{n_{jl}}\sum_{v=1}^{m_{jlk}^{(i)}}c(x-X_{jlkv}^{(i)})\big\}.
\end{equation}

Let $ \widehat{\textbf{p}}=(\widehat{p}_{11},\widehat{p}_{12},\cdots,\widehat{p}_{T1}, \widehat{p}_{T2})^\top$ be the estimator of the relative effect size vector defined by 
$$\widehat{\textbf{p}}=\int  \widehat{G}d \widehat{\textbf{F}},$$ 
where $ \widehat{\textbf{F}}=(\widehat{F}_{11},\widehat{F}_{12},\cdots,\widehat{F}_{T1}, \widehat{F}_{T2})$ is an estimator of $\textbf{F}$ and $ \widehat{G}=(2T)^{-1}\sum_{j=1}^{T}\sum_{l=1}^{2} \widehat{F}_{jl}$ is an estimator of $G$. Substituting the distribution functions with their empirical counterparts in (\ref{equation:wrsjl}), the relative effect size estimator can be taken as 
$$ \widehat{p}_{jl}= \widehat{\overline{w}}_{\cdot\cdot jl}=\frac{1}{2T}\sum_{r=1}^{T}\sum_{s=1}^{2} \widehat{w}_{rsjl}$$
and
\begin{equation}
	\label{equation:what}
	\widehat{w}_{rsjl}=\int  \widehat{F}_{rs}d \widehat{F}_{jl}=\frac{1}{N_{rs}+N_{jl}}(\overline{R}_{jl\cdot\cdot}^{(rs)}-\overline{R}_{rs\cdot\cdot}^{(jl)})+\frac{1}{2},
\end{equation}
where $R_{jlkv}^{(rs)}$ denotes the mid-rank of $X_{jlkv}$ among all $N_{rs}+N_{jl}$ observations within the two sample vectors $(\bm{X}_{jl1}^\top,\cdots,\bm{X}_{jln_j}^\top)$ and $(\bm{X}_{rs1}^\top,\cdots,\bm{X}_{rsn_r}^\top)$, and $\overline{R}_{jl\cdot\cdot}^{(rs)}=\frac{1}{N_{jl}}\sum_{k=1}^{n_j^{(c)}+n_{jl}}\sum_{v=1}^{m_{jlk}}R_{jlkv}^{(rs)}$ is the mean of ranks. Here, $X_{jlkv}$, $\bm{X}_{jlk}$ and $m_{jlk}$ are generic representations of $X_{jlkv}^{(A)}$, $\bm{X}_{jlk}^{(A)}$ and $m_{jlk}^{(A)}$, $A\in\{c,i\}$, respectively.

\subsection{Asymptotic Distribution}
\label{sec:Asymptotic Distribution}
The asymptotic behavior of $\widehat{\textbf{p}}$ can be studied similarly as \cite{cui2020}, which is a one-sample problem with two intervention periods.
In particular, define the vector 
\[\textbf{w}_{jl}=(w_{jl11},w_{jl12},\cdots,w_{jlT1},w_{jlT2})^\top\]
and the matrix
\[\textbf{W}=(\textbf{w}_{11},\textbf{w}_{12},\cdots\textbf{w}_{T1},\textbf{w}_{T2})\in \mathbb{R}^{2T\times2T}.\]
The empirical counterparts of $\textbf{w}_{jl}$ and $\textbf{W}$ are denoted by $\widehat{\textbf{w}}_{jl}$ and $\widehat{\textbf{W}}$, respectively. It can be shown that
\begin{equation}
	\label{equation:pbyW}
	\textbf{p}=\textbf{E}_{2T}\textrm{vec}(\textbf{W})\quad\textrm{and}\quad \widehat{\textbf{p}}=\textbf{E}_{2T}\textrm{vec}(\widehat{\textbf{W}}).
\end{equation}
Here, vec() denotes the matrix operator which stacks the columns of a matrix on top of each other and the matrix $\textbf{E}_{2T}$ is given by
\[\textbf{E}_{2T}=(2T)^{-1}\mathbbm{1}^\top_{2T}\otimes\textbf{I}_{2T},\]
where $\textbf{I}_{2T}$ denotes the $2T$-dimensional identity matrix, $\mathbbm{1}_{2T}$ denotes the $2T\times1$ vector of all 1's and $\otimes$ denotes the Kronecker product operation. We use the representation of $\widehat{\textbf{p}}$ and $\textbf{p}$ in (\ref{equation:pbyW}) and obtain, 
\begin{equation}
	\label{equation:pinw}
	\sqrt{N}(\widehat{\textbf{p}}-\textbf{p})=\textbf{E}_{2T}\cdot
	\sqrt{N}(\textrm{vec}(\widehat{\textbf{W}})-\textrm{vec}(\textbf{W})).
\end{equation}
Applying the asymptotic equivalence theorem for the components $\widehat{W}_{rsjl}-W_{rsjl}$ of $\widehat{\textbf{W}}-\textbf{W}$ \citep{cui2020}, the random vector $\sqrt{N}(\textrm{vec}(\widehat{\textbf{W}})-\textrm{vec}(\textbf{W}))$ has the same distribution as $\sqrt{N}\textbf{Z}$,
where $\textbf{Z}=(\textbf{Z}_{11}^\top,\textbf{Z}_{12}^\top,\cdots,\textbf{Z}_{T1}^\top,\textbf{Z}_{T2}^\top)^\top$, $\textbf{Z}_{jl}=(Z_{11jl},Z_{12jl},\cdots,Z_{T1jl},Z_{T2jl})^\top$ and
\begin{equation}
	\label{equation:ZinF}
	\begin{split}
		Z_{rsjl}&=\frac{1}{N_{jl}}\sum_{k=1}^{n_{j}^{(c)}}\sum_{v=1}^{m_{jlk}^{(c)}}F_{rs}(X_{jlkv}^{(c)})-\frac{1}{N_{rs}}\sum_{k=1}^{n_{r}^{(c)}}\sum_{v=1}^{m_{rsk}^{(c)}}F_{jl}(X_{rskv}^{(c)})\\
		&\qquad+\frac{1}{N_{jl}}\sum_{k=1}^{n_{jl}}\sum_{v=1}^{m_{jlk}^{(i)}}F_{rs}(X_{jlkv}^{(i)})-\frac{1}{N_{rs}}\sum_{k=1}^{n_{rs}}\sum_{v=1}^{m_{rsk}^{(i)}}F_{jl}(X_{rskv}^{(i)})+1-2w_{rsjl}.\\
	\end{split}
\end{equation}
Then $\sqrt{N}(\widehat{\textbf{p}}-\textbf{p})$ and $\sqrt{N}\textbf{E}_{2T}\textbf{Z}$ have the same distribution follows from (\ref{equation:pinw}).

To facilitate the derivation of the asymptotic distribution, define $$Y_{jlrskv}^{(A)}=F_{rs}(X_{jlkv}^{(A)}),\quad \overline{Y}_{jlrsk\cdot}^{(A)}=\frac{1}{m_{jlk}^{(A)}}\sum_{v=1}^{m_{jlk}^{(A)}}F_{rs}(X_{jlkv}^{(A)}),$$ $$\overline{Y}_{jlrs\cdot\cdot}^{(c)}=\frac{1}{N_{jl}^{(c)}}\sum_{k=1}^{n_{j}^{(c)}}\sum_{v=1}^{m_{jlk}^{(c)}}F_{rs}(X_{jlkv}^{(c)})\quad\text{and}\quad\overline{Y}_{jlrs\cdot\cdot}^{(i)}=\frac{1}{N_{jl}^{(i)}}\sum_{k=1}^{{n_{jl}}}\sum_{v=1}^{m_{jlk}^{(i)}}F_{rs}(X_{jlkv}^{(i)}).$$ 
Then (\ref{equation:ZinF}) can be expressed as
\begin{equation}
	\label{equation:ZinY}
	\begin{split}
		Z_{rsjl}&=\sum_{k=1}^{n_{j}^{(c)}}\bigg[\frac{1}{N_{jl}}m_{jlk}^{(c)}\overline{Y}_{jlrsk\cdot}^{(c)}-\frac{1}{N_{rs}}m_{rsk}^{(c)}\overline{Y}_{rsjlk\cdot}^{(c)}\bigg]+\frac{1}{N_{jl}}\sum_{k=1}^{n_{jl}}m_{jlk}^{(i)}\overline{Y}_{jlrsk\cdot}^{(i)}\\
		&\qquad-\frac{1}{N_{rs}}\sum_{k=1}^{n_{rs}}m_{rsk}^{(i)}\overline{Y}_{rsjlk\cdot}^{(i)}+1-2w_{rsjl}.
	\end{split}
\end{equation}
The proof of the asymptotic equivalence theorem follows from the proof of Theorem 4.2 in \cite{cui2020}. We see from (\ref{equation:ZinY}) that $Z_{rspq}$ is the sum of three independent terms apart from some constants. Each of these terms are sums of independent random variables. Therefore, we get the Central Limit Theorem stated in Theorem \ref{Thm:Distribution}.
\begin{theorem}
	\label{Thm:Distribution}
	Let $\textbf{V}=\textbf{E}_{2T}\textrm{Cov}(\sqrt{N}\textbf{Z})\textbf{E}_{2T}^\top$, then $\sqrt{N}( \widehat{\textbf{p}}-\textbf{p})$ is asymptotically multivariate normally distributed with expectation $\bm{0}$ and covariance matrix $\textbf{V}$ in (\ref{equation:ExplicitCov}).
\end{theorem}
\begin{proof}
	The asymptotic normality of $\sqrt{N}( \widehat{\textbf{p}}-\textbf{p})$ can be established from the asymptotic distribution of the random vector $\textbf{Z}$. Apart from some constants, $Z_{jlrs}$ is the sum of three independent random variables. Since the random variables $\overline{Y}_{jlrsk\cdot}^{(A)}$, $j=1,\cdots,T,l=1,2,A\in\{c,i\}$ are uniformly bounded by Assumption \ref{Assusmption:SampleSize}, asymptotic normality of $Z_{jlrs}$ can be asserted by verifying the Lindeberg's condition. Furthermore, the joint normality of $\textbf{Z}$ is established by the Cram\'er--Wold device. Finally, the multivariate delta method is applied to complete the proof.
	
\end{proof}

\subsection{Structure of the Covariance Matrix}
In this section, we present the explicit expression for the covariance matrix $\bm{V}$ in Theorem \ref{Thm:Distribution}. Note that
\begin{equation}
	\label{equation:ExplicitCov}
	\textbf{V}=N\textbf{E}_{2T}\bm{\Sigma}\textbf{E}_{2T}^\top,
\end{equation}
where $\bm{\Sigma}=\textrm{Cov}(\bm{\textbf{Z}})$. We partition $\bm{\Sigma}$ as $\bm{\Sigma}=(\bm{\Sigma}_{(jl)(rs)})_{1\le j,r\le T,1\le l,s,\le2}\in\mathbb{R}^{2T\times2T}$.
{Set 
	$$\bm{\Sigma}_{(jl)(rs)}=\textrm{Cov}(\textbf{Z}_{jl},\textbf{Z}_{rs})=(\textrm{Cov}(Z_{pqjl},Z_{p'q'rs}))_{1\le p,p'\le T,1\le q,q'\le 2}:=(\sigma_{jlrs}(p,q,p',q'))_{1\le p,p'\le T,1\le q,q'\le 2}.$$ 
As displayed in Appendix \ref{Appendix:CovarianceDecomposition}, each $\sigma_{jlrs}(p,q,p',q')$ can be decomposed into 16 variance components, i.e. $C_1,C_2,\cdots,C_{16}$. We realize the random variables $Z's$ have some favorable properties. First, $Z_{rsrs}=0$ since $w_{rsrs}=1/2$. Further, $Z_{rsjl}=-Z_{jlrs}$ since $w_{rsjl}=1-w_{jlrs}$. Then by independence of $X_{jlkv}^{(A)}$ and $X_{j'l'k'v'}^{(A')}$ for $A\ne A',A,A'\in\{c,i\}$ and for all $j\ne j'$ or $k\ne k'$, it follows that 
\[C_3=C_4=C_{7}=C_{8}=C_{9}=C_{10}=C_{13}=C_{14}=0\] 
so that 
\[\sigma_{jlrs}(p,q,p',q')=C_1-C_2-C_5+C_6+C_{11}-C_{12}-C_{15}+C_{16}.\]
The full expression of $\sigma_{jlrs}(p,q,p',q')$ is rather complicated, we instead investigate each of the variance components individually. Specifically,
	\begin{equation*}
		C_1=
		\begin{cases}
			0&p\ne p'\text{ or }\{p=p'=j\text{ and }l=q\}\text{ or }\{p=p'=r\text{ and } s=q'\}\\
			\tau_p^{(q,q')}(j,l,r,s)&otherwise\\
		\end{cases},
	\end{equation*}
	\begin{equation*}
		C_2=
		\begin{cases}
			0&p\ne r\text{ or }\{p=r=j\text{ and } l=q\}\text{ or }\{p=p'=r\text{ and } q'=s\}\\
			\tau_p^{(q,s)}(j,l,p',q')&otherwise\\
		\end{cases},
	\end{equation*}
	\begin{equation*}
		C_5=
		\begin{cases}
			0&j\ne p'\text{ or }\{p=p'=j\text{ and } q=l\}\text{ or }\{j=r=p'\text{ and } s=q'\}\\
			\tau_j^{(l,q')}(p,q,r,s)&otherwise\\
		\end{cases},
	\end{equation*}
	\begin{equation*}
		C_6=
		\begin{cases}
			0&j\ne r\text{ or }\{p=r=j\text{ and } q=l\}\text{ or }\{j=r=p'\text{ and } q'=s\}\\
			\tau_j^{(l,s)}(p,q,p',q')&otherwise
		\end{cases},
	\end{equation*}
	\begin{equation*}
		C_{11}=
		\begin{cases}
			0&p\ne p'\text{ or } q\ne q'\text{ or }\{p=p'=j\text{ and } q=q'=l\}\text{ or }\{p=p'=r\text{ and } q=q'=s\}\\
			\eta_p^{(q,q')}(j,l,r,s)&otherwise\\
		\end{cases},
	\end{equation*}
	\begin{equation*}
		C_{12}=
		\begin{cases}
			0&p\ne r\text{ or }q\ne s\text{ or }\{p=r=j\text{ and } q=s=l\}\text{ or }\{p=p'=r\text{ and } q=q'=s\}\\
			\eta_p^{(q,s)}(j,l,p',q')&otherwise\\
		\end{cases},
	\end{equation*}
	\begin{equation*}
		C_{15}=
		\begin{cases}
			0&j\ne p'\text{ or }l\ne q'\text{ or }\{p=p'=j\text{ and } q=q'=l\}\text{ or }\{j=r=p'\text{ and } l=s=q'\}\\
			\eta_j^{(l,q')}(p,q,r,s)&otherwise\\
		\end{cases},
	\end{equation*}
and
	\begin{equation*}
		C_{16}=
		\begin{cases}
			0&j\ne r\text{ or }l\ne s\text{ or }\{p=r=j\text{ and } l=s=q\}\text{ or }\{j=r=p'\text{ and } l=s=q'\}\\
			\eta_j^{(l,s)}(p,q,p',q')&otherwise
		\end{cases},
	\end{equation*}
where \[\tau_r^{(s,l)}(p,q,p',q')=\frac{1}{N_{rs}N_{rl}}Cov\bigg(\sum_{k=1}^{n_{r}^{(c)}}\sum_{v=1}^{m_{rsk}^{(c)}}F_{pq}(X_{rskv}^{(c)}),\sum_{k=1}^{n_{r}^{(c)}}\sum_{v=1}^{m_{rlk}^{(c)}}F_{p'q'}(X_{rlkv}^{(c)})\bigg)\]
and
\[\eta_r^{(s,l)}(p,q,p',q')=\frac{1}{N_{s}N_{l}}Cov\bigg(\sum_{k=1}^{n_{rs}}\sum_{v=1}^{m_{rsk}^{(i)}}F_{pq}(X_{rskv}^{(i)}),\sum_{k=1}^{n_{rl}}\sum_{v=1}^{m_{rlk}^{(i)}}F_{p'q'}(X_{rlkv}^{(i)})\bigg).\]
}

\subsection{Estimation of the Asymptotic Covariance Matrix}
\label{subsec:EstimateCov}
In the previous section, we derived the exact expression of $\sigma_{jlrs}(p,q,p',q')$ in terms of $\tau's$ and $\eta's$. Therefore, for the purpose of estimating $\textbf{V}$, we need to estimate $\tau's$ and $\eta's$. 
Due to the independence among clusters, it follows $\eta_r^{(s,l)}(p,q,p',q')=0$ if $s\ne l$. Then, we can rewrite
\begin{equation}
	\label{equation:tau}
	\begin{split}
		\tau_r^{(s,l)}(p,q,p',q')
		&=\frac{1}{N_{rs}N_{rl}}\sum_{k=1}^{n_{r}^{(c)}}Cov\bigg(m_{rsk}^{(c)}\overline{Y}_{rspqk\cdot}^{(c)},m_{rlk}^{(c)}\overline{Y}_{rlp'q'k\cdot}^{(c)}\bigg)\\
		&=\frac{n_{r}^{(c)}}{N_{rs}N_{rl}}\frac{1}{n_{r}^{(c)}}\sum_{k=1}^{n_{r}^{(c)}}Cov\bigg(m_{rsk}^{(c)}\overline{Y}_{rspqk\cdot}^{(c)},m_{rlk}^{(c)}\overline{Y}_{rlp'q'k\cdot}^{(c)}\bigg)\\
		&=\frac{n_{r}^{(c)}}{N_{rs}N_{rl}}\zeta^{2(c)}_{r(s,l)}(p,q,p',q')\\
	\end{split}
\end{equation} 
and
\begin{equation}
	\label{equation:eta}
	\begin{split}
		\eta_r^{(s,s)}(p,q,p',q')
		&=\frac{1}{N_{rs}^2}Cov\bigg(\sum_{k=1}^{n_{rs}}m_{rsk}^{(i)}\overline{Y}_{rspqk\cdot}^{(i)},\sum_{k'=1}^{n_{rs}}m_{rsk}^{(i)}\overline{Y}_{rsp'q'k\cdot}^{(i)}\bigg)\\
		&=\frac{1}{N_{rs}^2}\sum_{k=1}^{n_{rs}}Cov\bigg(m_{rsk}^{(i)}\overline{Y}_{rspqk\cdot}^{(i)},m_{rsk}^{(i)}\overline{Y}_{rsp'q'k\cdot}^{(i)}\bigg)\\
		&=\frac{n_{rs}}{N_{rs}^2}\frac{1}{n_{rs}}\sum_{k=1}^{n_{rs}}Cov\bigg(m_{rsk}^{(i)}\overline{Y}_{rspqk\cdot}^{(i)},m_{rsk}^{(i)}\overline{Y}_{rsp'q'k\cdot}^{(i)}\bigg)\\
		&=\frac{n_{rs}}{N_{rs}^2}\zeta^{2(i)}_{r(s,s)}(p,q,p',q'),\\
	\end{split}
\end{equation} 
where $$\zeta^{2(c)}_{r(s,l)}(p,q,p',q')=\frac{1}{n_{r}^{(c)}}\sum_{k=1}^{n_{r}^{(c)}}Cov\big(m_{rsk}^{(c)}\overline{Y}_{rspqk\cdot}^{(c)},m_{rlk}^{(c)}\overline{Y}_{rlp'q'k\cdot}^{(c)}\big)$$ and $$\zeta^{2(i)}_{r(s,s)}(p,q,p',q')=\frac{1}{n_{rs}}\sum_{k=1}^{n_{rs}}Cov\big(m_{rsk}^{(i)}\overline{Y}_{rspqk\cdot}^{(i)},m_{rsk}^{(i)}\overline{Y}_{rsp'q'k\cdot}^{(i)}\big).$$ 

If $\overline{Y}_{rspqk\cdot}^{(c)}$ and $\overline{Y}_{rspqk\cdot}^{(i)}$ were observable and their expected values were known, then natural estimators of $\zeta^{2(c)}_{r(s,l)}(p,q,p',q')$ and $\zeta^{2(i)}_{r(s,s)}(p,q,p',q')$ would be
\[\tilde{\zeta}^{2(c)}_{r(s,l)}(p,q,p',q')=\frac{1}{n_{r}^{(c)}}\sum_{k=1}^{n_{r}^{(c)}}\left[m_{rsk}^{(c)}\overline{Y}_{rspqk\cdot}^{(c)}-E(m_{rsk}^{(c)}\overline{Y}_{rspqk\cdot}^{(c)})\right]\left[m_{rlk}^{(c)}\overline{Y}_{rlp'q'k\cdot}^{(c)}-E(m_{rlk}^{(c)}\overline{Y}_{rlp'q'k\cdot}^{(c)})\right]\]
and
\[\tilde{\zeta}^{2(i)}_{r(s,s)}(p,q,p',q')=\frac{1}{n_{rs}}\sum_{k=1}^{n_{rs}}\left[m_{rsk}^{(i)}\overline{Y}_{rspqk\cdot}^{(i)}-E(m_{rsk}^{(i)}\overline{Y}_{rspqk\cdot}^{(i)})\right]\left[m_{rsk}^{(i)}\overline{Y}_{rsp'q'k\cdot}^{(i)}-E(m_{rsk}^{(i)}\overline{Y}_{rsp'q'k\cdot}^{(i)})\right].\]
It is easy to verify that $E(m_{rsk}^{(A)}\overline{Y}_{rspqk\cdot}^{(A)})=m_{rsk}^{(A)}w_{pqrs}$ for $A\in\{c,i\}$. Consequently, 
\[\tilde{\zeta}^{2(c)}_{r(s,l)}(p,q,p',q')=\frac{1}{n_{r}^{(c)}}\sum_{k=1}^{n_{r}^{(c)}}m_{rsk}^{(c)}m_{rlk}^{(c)}\left(\overline{Y}_{rspqk\cdot}^{(c)}-w_{pqrs}\right)\left(\overline{Y}_{rlp'q'k\cdot}^{(c)}-w_{p'q'rl}\right)\]
and
\[\tilde{\zeta}^{2(i)}_{r(s,s)}(p,q,p',q')=\frac{1}{n_{rs}}\sum_{k=1}^{n_{rs}}m_{rsk}^{2(i)}\left(\overline{Y}_{rspqk\cdot}^{(i)}-w_{pqrs}\right)\left(\overline{Y}_{rsp'q'k\cdot}^{(i)}-w_{p'q'rs}\right).\]
To construct estimators of the covariances that are computable, we replace the unobservable random variables with observable ones that are asymptotically "close" in probability sense. Therefore, we substitute the distribution functions $F_{rs}$ with their empirical counterparts $ \widehat{F}_{rs}$ in (\ref{equation:empirical}) for all $r=1,\cdots,T$ and $s=1,2$. More specifically, let
$$\widehat{\overline{Y}}_{jlrsk\cdot}^{(A)}=\frac{1}{m_{jlk}^{(A)}}\sum_{v=1}^{m_{jlk}^{(A)}} \widehat{F}_{rs}(X_{jlkv}^{(A)}),\quad j,s\in\{1,\cdots,T\},l,s\in\{1,2\}\text{ and }A\in\{c,i\}$$
be the empirical counterpart of $\overline{Y}_{jlrsk\cdot}^{(A)}$. Naturally, expectations of the unobservable random variables are estimated by replacing $w's$ with $ \widehat{w}'s$ in (\ref{equation:what}). Therefore, reasonable estimators of $\tau's$ and $\eta's$ are
\begin{equation}
	\label{equation:tauhat}
	\widehat{\tau}_r^{(s,l)}(p,q,p',q')=\frac{n_{r}^{(c)}}{N_{rs}N_{rl}}\frac{1}{n_{r}^{(c)}-1}\sum_{k=1}^{n_{r}^{(c)}}m_{rsk}^{(c)}m_{rlk}^{(c)}\left( \widehat{\overline{Y}}_{rspqk\cdot}^{(c)}-\widehat{w}_{pqrs}\right)\left( \widehat{\overline{Y}}_{rlp'q'k\cdot}^{(c)}-\widehat{w}_{p'q'rl}\right)
\end{equation}
and 
\begin{equation}
	\label{equation:etahat}
	\widehat{\eta}^{(s,s)}_{r}(p,q,p',q')=\frac{n_{rs}}{N_{rs}^2}\frac{1}{n_{rs}-1}\sum_{k=1}^{n_{rs}}m_{rsk}^{2(i)}\left(\widehat{\overline{Y}}_{rspqk\cdot}^{(i)}-\widehat{w}_{pqrs}\right)\left(\widehat{\overline{Y}}_{rsp'q'k\cdot}^{(i)}-\widehat{w}_{p'q'rs}\right).
\end{equation}
Replacing the quantities $\tau's$ and $\eta's$ with their estimators $ \widehat{\tau}'s$ and $ \widehat{\eta}'s$ in (\ref{equation:tauhat}) and (\ref{equation:etahat}), we obtain estimators of $\widehat{\sigma}_{rs}(p,q,p',q')$ and $\widehat{\sigma}_{jlrs}(p,q,p',q')$. The resulting estimator of the asymptotic covariance matrix $\bm{\Sigma}$ of $\textbf{Z}$ is denoted by $\widehat{\bm{\Sigma}}_N$=$(\widehat{\bm{\Sigma}}_{(jl)(rs)})_{j,r=1,\cdots,T,l,s=1,2}$. Finally, from equation (\ref{equation:ExplicitCov}), we obtain a consistent estimator $\widehat{\textbf{V}}_N$ of the asymptotic covariance matrix $\textbf{V}$ of $\sqrt{N}(\widehat{\textbf{p}}-\textbf{p})$
\begin{equation}
	\label{equation:ElementInVN}
	\widehat{\textbf{V}}_N=N\textbf{E}_{2T}\widehat{\bm{\Sigma}}_N\textbf{E}_{2T}^\top=(\widehat{v}_{(jl)(rs)})_{1\le j,r\le T,1\le l,s\le2},
\end{equation}
where $\widehat{v}_{(jl)(rs)}=(2T)^{-2}\mathbbm{1}_{2T}^\top\widehat{\bm{\Sigma}}_{(jl)(rs)}\mathbbm
{1}_{2T}$. Consistency of $\widehat{\textbf{V}}_N$ is established in Theorem \ref{Thm:ConsistencyCov}. 
\begin{theorem}
	\label{Thm:ConsistencyCov}
	Under assumptions \ref{Assusmption:SampleSize} and \ref{Assumption:ClusterSize},
	\[\parallel \widehat{\bm{V}}_N-\bm{V}\parallel_2^2\rightarrow0.\]
\end{theorem}
\begin{proof}
	To prove the consistency of $\widehat{\bm{V}}_N$, it suffices to prove the consistency of $ \widehat{\tau}_r^{(s,l)}(p,q,p',q')$ and $ \widehat{\eta}^{(s,s)}_{r}(p,q,p',q')$. By Assumption \ref{Assusmption:SampleSize}, it is enough to show consistency of $ \widehat{\zeta}^{2(c)}_{r(s,l)}(p,q,p',q')$ and $ \widehat{\zeta}^{2(i)}_{r(s,s)}(p,q,p',q')$. 
	Refer to the proof of Theorem 4.3 in \cite{cui2020} for special cases, i.e. $(s,p,q)=(l,p',q')$ for $ \widehat{\zeta}^{2(c)}_{r(s,l)}(p,q,p',q')$ and $(p,q)=(p',q')$ for $ \widehat{\zeta}^{2(i)}_{r(s,s)}(p,q,p',q')$. The other cases can be proved analogously.
\end{proof}

\section{Test Statistics and Confidence Interval}
\label{sec:TestStat}
In Theorem \ref{Thm:Distribution}, we showed that $\sqrt{N}( \widehat{\textbf{p}}-\textbf{p})$ is asymptotically normally distributed with mean $\bm{0}$ and covariance matrix $\textbf{V}$. This result allows construction of approximate test procedures for the null hypothesis $H_0^p:\textbf{Cp}=\bm{0}$ introduced in Section \ref{sec:HypothesisSetting}. 

Let $\textbf{M}^+$ denote the Moore-Penrose inverse of a matrix $\textbf{M}$. For testing $H_0^p$, the Wald-type statistic (WTS)
\begin{equation}
	\label{WTS}
	W_N(\textbf{C})=N\widehat{\textbf{p}}^\top\textbf{C}^\top(\textbf{C}\widehat{\bm{V}}_N\textbf{C}^\top)^+\textbf{C}\widehat{\textbf{p}}
\end{equation}
may be utilized. However, as discussed in \cite{ANOVA-1997}, \cite{vallejo2010analysis}, \cite{umlauft2017rank} and \cite{brunner2017rank}, $W_N(\textbf{C})$ suffers from poor control of Type-I error (liberal performance) under small sample sizes. This liberal property is even worse in our case because of the complicated structure of the covariance matrix $\textbf{V}$ that involves a large number of unknown quantities which need to be estimated. {See also \cite{brunner2017rank} for  other problems associated with this test.}  

Another test that is known to perform well in small sample is the ANOVA-type statistic 
\begin{equation}
	\label{equation:Q_N}
	Q_N(\textbf{C})=Q_N(\textbf{T})=\frac{N}{\text{tr}(\textbf{T}\widehat{\textbf{V}}_N)}\widehat{\textbf{p}}'\textbf{T}\widehat{\textbf{p}},
\end{equation}
where $\textbf{T}$ is given in (\ref{equation:T}). Specific test statistics can be derived by plugging in $\textbf{T}_R$ for $\textbf{T}$ in (\ref{equation:Q_N}), where $\textbf{T}_R$, $R\in\{A,B,AB\}$ are given Section \ref{sec:HypothesisSetting}. Note that $H_0^p:\textbf{Tp=0}$ is equivalent to $H_0^p:\textbf{Cp}=\textbf{0}$ since $\textbf{C}^\top(\textbf{C}\textbf{C}^\top)^+$ is a generalized inverse of $\textbf{C}$. 

Following \cite{brunner2017rank}, we consider the ANOVA-Box-type $\chi^2$ approximation procedure where the distribution of $Q_N(\textbf{T})$ is approximated by Box-type approximation (see \cite{Box-1954} and \cite{ANOVA-1997}) as
\begin{equation}
	\label{BoxApproximation}
	\chi_f^2/f,
\end{equation}
where $f$ is estimated by 
\begin{equation}
	\label{fhat}
	\widehat{f}=\frac{\text{tr}^2(\textbf{T}\widehat{\textbf{V}}_N)}{\text{tr}(\textbf{T}\widehat{\textbf{V}}_N\textbf{T}\widehat{\textbf{V}}_N)}.
\end{equation}

Based on the aforementioned approximation procedure, the asymptotic level $\alpha$ test would reject  $H_0^p$ if 
	$\widehat{f}Q_N(\textbf{T})>\chi^2_{\widehat{f},1-\alpha}$ where $\chi^2_{\widehat{f},1-\alpha}$ is the $(1-\alpha)$-quantile of a Chi-squared distribution with $\widehat{f}$ degrees of freedom.
Note $\phi_N$ is consistent for alternatives, see Theorem 5 in \cite{brunner2017rank} for more details.

It is straightforward to derive confidence intervals for the nonparametric effect size $p_{jl}$ in (\ref{equation:p_{jl}}) and contrasts of them from the asymptotic results in Section \ref{sec:AsymptoticTheory}. However, it is well known that these confidence intervals may not be range preserving. As recommended and discussed by a number of papers (see, for example, \cite{matched-pair-2012} and \cite{cui2020}), logit or probit transformations can be used to achieve ranging preserving intervals. For example, an approximate $(1-\alpha)$ confidence interval for $p_{jl}$ can be obtained as the inverse image (with respect to $g$) of the interval  $\widehat{p}_{jl}\pm\frac{z_{1-\alpha/2}}{\sqrt{N}}\sqrt{\widehat{v}_{(jl)(jl)}}g'(\widehat{p}_{jl})$
for any given function $g(\cdot)$ that is differentiable in $p_{jl}$ with $g'(p_{jl})\ne0$, where $\widehat{v}_{(jl)(jl)}$ is given in (\ref{equation:ElementInVN}).

\section{Simulation Study}
\label{sec:SimuStudy}
In this section, we evaluate the small-sample properties of the statistical test ${\phi}_N$ based on the ANOVA-type statistic $Q_N(\textbf
T)$ in equation (\ref{equation:Q_N}). In a small-scale simulation we investigate (a) maintenance of the preassigned Type-I error level ($\alpha=5\%$) and
(b) powers achieved to detect specific alternatives.

\subsection{Simulation Design}
The simulation studies seek to answer questions (a)  and (b) above in a variety of scenarios that cover a set of reasonable model involving strongly- and weakly-correlated clustered data with small and moderate cluster and sample sizes. For brevity, we consider same sample size allocations in all intervention groups, i.e. $(n_j^{(c)},n_{j1},n_{j2})=(n_c,n_{1},n_{2})$ are the same for all $j\in\{1,\cdots,T\}$. Within each intervention group and at each intervention period, the cluster sizes are generated from \textit{Binomial(2,0.3)} and \textit{Binomial(5,0.3)} and 1 is added to the generated numbers to avoid 0. Therefore, the maximum cluster sizes $M$ are 3 and 6 for all the settings. The shape of the data distribution is another feature we investigate in evaluating the performance of the test procedures. Three different multivariate distributions are used to generate clustered data; namely, Discretized Multivariate Normal, Multivariate Log-Normal and Multivariate Cauchy. Covariance or scale matrices of data generated for the $k^{th}$ cluster in the $j^{th}$ intervention group and during the $l^{th}$ intervention period are set to the structure
\begin{equation}
	\label{equation:CovMatrix}
	\begin{bmatrix}   
		\sigma_{1}^2\textbf{\textrm{I}}_{m_{j1k}}+\rho_{1}\sigma_{1}^2(\textbf{\textrm{J}}_{m_{j1k}}-\textbf{\textrm{I}}_{m_{j1k}})&\rho_{12}\sigma_{1}\sigma_{2}\textbf{\textrm{J}}_{m_{j2k}}\\
		\rho_{12}\sigma_{1}\sigma_{2}\textbf{\textrm{J}}_{m_{j2k}}&\sigma_{2}^2\textbf{\textrm{I}}_{m_{j1k}}+\rho_{2}\sigma_{2}^2(\textbf{\textrm{J}}_{m_{j2k}}-\textbf{\textrm{I}}_{m_{j2k}})    
	\end{bmatrix},
\end{equation}
where $\sigma_{l}^2$ is the variance and $\rho_{l}$ and $\rho_{12}$ are inter-cluster and intra-cluster correlation coefficients. In particular, covariance matrices are set to be the same across all intervention groups. 
The impact of strong and weak intra-cluster and inter-cluster correlations will be investigated by varying the values of $\rho_1$, $\rho_2$ and $\rho_{12}$. 
Further, homoscedastic and heteroscedastic scenarios are covered by settings $\sigma_1^2=\sigma_2^2$ and $\sigma_1^2\ne\sigma_2^2$, respectively.  
We consider a pre-post design with only 3 intervention groups due to time complexity. The various simulation settings are summarized in Table \ref{Table:SimuSetting}. All the computations are done in the \texttt{R} environment (version 3.6.0) and the run size is 1000. 

\begin{table}[!htb]
	\centering
	\caption{Simulation Settings }
	\label{Table:SimuSetting}
	\begin{tabular}{|c|c|}
		\hline
		\multirow{3}{*}{Distributions} & \multirow{2}{*}{} \\
		&  Discretized Multivariate Normal, Multivariate Log-normal,\\ 
		&  Multivariate Cauchy\\ \hline
		Sample sizes $(n_c,n_1,n_2)$ & (5,10,5), (10,5,5) \\ \hline
		Maximum Cluster Size $M$ & 3,6 \\ \hline
		Correlation $(\rho_1,\rho_2,\rho_{12})$ & (0.9,0.9,0.1), (0.1,0.1,0.9), (0.1,0.9,0.9) \\ \hline
		Variance $(\sigma_1^2,\sigma_2^2)$ & (1,1), (1,5) \\ \hline
	\end{tabular}
\end{table}

\subsection{Size Simulation}
The achieved Type-I error rates for tests on intervention, time and their interaction effects are displayed in Table \ref{Table:TypeOne}. In summary, the achieved Type-I error rates are not quite influenced by sample sizes, correlation coefficients, variances or cluster sizes. Also, they do not significantly differ by distributions. Most of the achieved Type-I error rates are close to the nominal significance level $\alpha=0.05$. Overall, the proposed test  preserves the preassigned significance level well for intervention, time and the interaction effects well under various situations.

\begin{table}[htb]
	\centering
	\caption{ Achieved Type-I error rate ($\times100$) for the ANOVA-type tests on intervention (I), time (T) and interaction (I$\times$T) effects based on the $\chi^2$ approximation. Data are from discretized multivariate normal distribution, multivariate log-normal distribution, and multivariate Cauchy distribution. The nominal Type-I error rate is $\alpha=0.05$.}
	\begin{tabular}{cccc|c|c|c|c|c|c|c|c|c|}
		\cline{5-13}
		&                                                     &                                                &     & \multicolumn{3}{c|}{\begin{tabular}[c]{@{}c@{}}Discretized\\ Multivariate\\ Normal\\ Distribution\end{tabular}} & \multicolumn{3}{c|}{\begin{tabular}[c]{@{}c@{}}Multivariate\\ Log-Normal\\ Distribution\end{tabular}} & \multicolumn{3}{c|}{\begin{tabular}[c]{@{}c@{}}Multivariate\\ Cauchy\\ Distribution\end{tabular}} \\ \hline
		\multicolumn{1}{|c|}{$(n_c,n_1,n_2)$}            & \multicolumn{1}{c|}{$(\rho_1,\rho_2,\rho_{12})$}    & \multicolumn{1}{c|}{$(\sigma_1^2,\sigma_2^2)$} & $M$ & I  & T & I$\times$T  & I & T  & I$\times$T & I & T & I$\times$T \\ \hline
		\multicolumn{1}{|c|}{\multirow{12}{*}{(5,10,5)}} & \multicolumn{1}{c|}{\multirow{4}{*}{(0.9,0.9,0.1)}} & \multicolumn{1}{c|}{\multirow{2}{*}{(1,1)}}    & 3   & 5.6 & 5.0  & 5.4 & 5.3 & 5.0  & 6.4 & 5.6 & 5.9 & 4.5 \\ \cline{4-13} 
		\multicolumn{1}{|c|}{}                           & \multicolumn{1}{c|}{}                               & \multicolumn{1}{c|}{}                          & 6   & 6.2 & 5.1 & 5.0   & 6.3 & 5.5 & 4.9 & 6.2 & 4.2 & 5.6 \\ \cline{3-13} 
		\multicolumn{1}{|c|}{}                           & \multicolumn{1}{c|}{}                               & \multicolumn{1}{c|}{\multirow{2}{*}{(1,5)}}    & 3   & 5.5 & 5.9 & 6.5 & 5.0  & 5.7 & 6.5 & 5.7 & 6.9 & 4.7 \\ \cline{4-13} 
		\multicolumn{1}{|c|}{}                           & \multicolumn{1}{c|}{}                               & \multicolumn{1}{c|}{}                          & 6   & 6.1 & 5.6 & 5.2 & 6.6 & 5.9 & 4.8 & 6.7 & 4.8 & 5.3 \\ \cline{2-13} 
		\multicolumn{1}{|c|}{}                           & \multicolumn{1}{c|}{\multirow{4}{*}{(0.1,0.1,0.9)}} & \multicolumn{1}{c|}{\multirow{2}{*}{(1,1)}}    & 3   & 4.6 & 4.5 & 4.4 & 4.4 & 5.5 & 4.1 & 4.6 & 4.1 & 4.3 \\ \cline{4-13} 
		\multicolumn{1}{|c|}{}                           & \multicolumn{1}{c|}{}                               & \multicolumn{1}{c|}{}                          & 6   & 5.8 & 4.4 & 4.7 & 5.8 & 4.7 & 4.8 & 5.8 & 5.0   & 5.8 \\ \cline{3-13} 
		\multicolumn{1}{|c|}{}                           & \multicolumn{1}{c|}{}                               & \multicolumn{1}{c|}{\multirow{2}{*}{(1,5)}}    & 3   & 5.2 & 5.0  & 4.6 & 4.4 & 4.6 & 4.3 & 4.9 & 4.8 & 4.6 \\ \cline{4-13} 
		\multicolumn{1}{|c|}{}                           & \multicolumn{1}{c|}{}                               & \multicolumn{1}{c|}{}                          & 6   & 6.2 & 4.5 & 5.4 & 6.0  & 4.4 & 4.9 & 6.0   & 5.9 & 5.2 \\ \cline{2-13} 
		\multicolumn{1}{|c|}{}        &               \multicolumn{1}{c|}{\multirow{4}{*}{(0.1,0.9,0.9)}} & \multicolumn{1}{c|}{\multirow{2}{*}{(1,1)}}    & 3   & 4.4 & 3.9 & 5.0   & 5.0   & 4.5 & 5.7 & 5.7 & 6.0   & 5.2  \\ \cline{4-13} 
		\multicolumn{1}{|c|}{}                           & \multicolumn{1}{c|}{}                               & \multicolumn{1}{c|}{}                          & 6   & 6.2 & 4.2 & 4.7 & 6.0   & 4.7 & 4.8 & 4.5 & 5.9 & 5.0  \\ \cline{3-13} 
		\multicolumn{1}{|c|}{}                           & \multicolumn{1}{c|}{}                               & \multicolumn{1}{c|}{\multirow{2}{*}{(1,5)}}    & 3   & 6.6 & 4.4 & 6.9 & 6.6 & 4.2 & 7.0   & 5.1 & 5.9 & 5.0  \\ \cline{4-13} 
		\multicolumn{1}{|c|}{}                           & \multicolumn{1}{c|}{}                               & \multicolumn{1}{c|}{}                          & 6   &  6.2 & 5.2 & 5.3 & 6.4 & 5.0   & 5.5 & 5.4 & 5.3 & 4.8 \\ \hline
		\multicolumn{1}{|c|}{\multirow{12}{*}{(10,5,5)}} & \multicolumn{1}{c|}{\multirow{4}{*}{(0.9,0.9,0.1)}} & \multicolumn{1}{c|}{\multirow{2}{*}{(1,1)}}    & 3   & 5.6 & 4.9 & 5.3 & 5.1 & 5.9 & 5.9 & 5.6 & 5.7 & 4.9 \\ \cline{4-13} 
		\multicolumn{1}{|c|}{}                           & \multicolumn{1}{c|}{}                               & \multicolumn{1}{c|}{}                          & 6   &  5.0   & 6.0   & 5.9 & 5.2 & 6.2 & 5.4 & 4.1 & 5.7 & 4.9 \\ \cline{3-13} 
		\multicolumn{1}{|c|}{}                           & \multicolumn{1}{c|}{}                               & \multicolumn{1}{c|}{\multirow{2}{*}{(1,5)}}    & 3   & 4.7 & 5.9 & 5.2 & 4.9 & 5.7 & 5.3 & 5.7 & 6.0   & 5.3 \\ \cline{4-13} 
		\multicolumn{1}{|c|}{}                           & \multicolumn{1}{c|}{}                               & \multicolumn{1}{c|}{}                          & 6   &  5.3 & 5.9 & 5.4 & 5.8 & 5.9 & 5.3 & 4.8 & 4.9 & 4.4 \\ \cline{2-13} 
		\multicolumn{1}{|c|}{}                           & \multicolumn{1}{c|}{\multirow{4}{*}{(0.1,0.1,0.9)}} & \multicolumn{1}{c|}{\multirow{2}{*}{(1,1)}}    & 3   & 6.2 & 4.8 & 3.9 & 6.7 & 4.8 & 4.0   & 5.2 & 4.8 & 3.1 \\ \cline{4-13} 
		\multicolumn{1}{|c|}{}                           & \multicolumn{1}{c|}{}                               & \multicolumn{1}{c|}{}                          & 6   & 4.3 & 5.6 & 5.1 & 4.4 & 5.4 & 4.6 & 5.5 & 5.2 & 5.7 \\ \cline{3-13} 
		\multicolumn{1}{|c|}{}                           & \multicolumn{1}{c|}{}                               & \multicolumn{1}{c|}{\multirow{2}{*}{(1,5)}}    & 3   & 5.1 & 4.4 & 3.4 & 4.6 & 4.6 & 3.4 & 6.0   & 4.7 & 4.3 \\ \cline{4-13} 
		\multicolumn{1}{|c|}{}                           & \multicolumn{1}{c|}{}                               & \multicolumn{1}{c|}{}                          & 6   & 4.7 & 5.6 & 4.8 & 4.4 & 5.3 & 4.9 & 6.0  & 5.1 & 4.8 \\ \cline{2-13} 
		\multicolumn{1}{|c|}{}        &               \multicolumn{1}{c|}{\multirow{4}{*}{(0.1,0.9,0.9)}} & \multicolumn{1}{c|}{\multirow{2}{*}{(1,1)}}    & 3   &  6.1 & 3.9 & 4.9 & 5.9 & 3.4 & 4.0   & 6.4 & 3.5 & 4.2 \\ \cline{4-13} 
		\multicolumn{1}{|c|}{}                           & \multicolumn{1}{c|}{}                               & \multicolumn{1}{c|}{}                          & 6   &  6.1 & 4.5 & 3.9 & 5.6 & 4.3 & 3.8 & 6.4 & 6.0   & 4.9 \\ \cline{3-13} 
		\multicolumn{1}{|c|}{}                           & \multicolumn{1}{c|}{}                               & \multicolumn{1}{c|}{\multirow{2}{*}{(1,5)}}    & 3   &  5.7 & 3.9 & 3.5 & 5.7 & 4.1 & 3.7 & 5.6 & 4.2 & 3.7\\ \cline{4-13} 
		\multicolumn{1}{|c|}{}                           & \multicolumn{1}{c|}{}                               & \multicolumn{1}{c|}{}                          & 6   &  5.5 & 5.2 & 5.0   & 6.2 & 5.4 & 4.7 & 5.9 & 5.6 & 4.3\\ \hline
	\end{tabular}
	\label{Table:TypeOne}
\end{table}

\subsection{Power Simulation}
To investigate the powers of the test  in detecting certain alternatives, we consider a balanced pre-post design with 3 intervention groups, where $n_c=n_1=n_2=5$ for all intervention groups. Other simulation settings are the same as displayed in Table \ref{Table:SimuSetting} except the variances. Since the powers in the  homogeneous  and heterogeneity cases are not quite different, only results for homogeneous data ($\sigma_1^2=\sigma_2^2=1$) are displayed in this section. Data are generated from three multivariate distributions with mean vector $\bm{\mu}_{jl}$ and covariance matrix as shown in (\ref{equation:CovMatrix}).  We consider three types of shift alternatives:
\begin{enumerate}
	\item[(a)] one-point alternative: $\bm{\mu}_{11}=\bm{\mu}_{12}=\bm{\mu}_{21}=\bm{\mu}_{22}=\bm{\mu}_{31}=\bm{0}$ and $\bm{\mu}_{32}=\delta\mathbbm{1}$;
	\item[(b)] one-time alternative: $\bm{\mu}_{11}=\bm{\mu}_{21}=\bm{\mu}_{31}=\bm{0}$ and $\bm{\mu}_{12}=\bm{\mu}_{22}=\bm{\mu}_{32}=\delta\mathbbm{1}$;
	\item[(c)] increasing-trend alternative: $\bm{\mu}_{11}=\delta/6\mathbbm{1}$, $\bm{\mu}_{12}=\delta/3\mathbbm{1}$, $\bm{\mu}_{21}=\delta/2\mathbbm{1}$, 
	$\bm{\mu}_{22}=2\delta/3\mathbbm{1}$,
	$\bm{\mu}_{31}=5\delta/6\mathbbm{1}$ and $\bm{\mu}_{32}=\delta\mathbbm{1}$.
\end{enumerate}
The  shift alternatives  target differnt aspects of departure from the null. More precisely, the one-point alternative (a) represents the case where there is intervention effect in only one of the intervention groups so that the intervention and time effects exist accordingly. The one-time alternative (b) represents the situation where there is no intervention or interaction effect but time effects exist in all the intervention groups.  The increasing-trend alternative (c) represents the special case where there is no interaction effect but intervention and time effects exist. In all the alternatives (a)-(c), $\delta$ is varied from 0 to 3 in steps of 0.3. The simulation results are displayed in Tables \ref{Table:Power-OnePoint-Homo-Normal}--\ref{Table:Power-Cauchy}.

First of all, for data that are generated from discretized multivariate normal distribution, larger $M$ values tend to produce larger powers while other settung remaining the same. For the one-point alternative, powers of the tests for all the three effects increase persistently with $\delta$. However, the rates of increase  are affected by the correlation coefficients. For example, when $(\rho_{1},\rho_{2},\rho_{12})=(0.9,0.9,0.1)$, the powers for intervention and time effects increase at a similar rate but they do so  faster than the intervention effect. Whereas for $(\rho_{1},\rho_{2},\rho_{12})=(0.1,0.1,0.9)$ or $(0.1,0.9,0.9)$, the rate of  increase in powers for all three effects are similar, and the comparison of the powers reveals that  Interaction$>$Time$>$Intervention. For the one-time alternative, powers of the time effect increase at much faster rate compared to the one-point alternative, and the rate of increase are again related to the correlation coefficients. In particular, powers for $(\rho_{1},\rho_{2},\rho_{12})=(0.1,0.1,0.9)$ or $(0.1,0.9,0.9)$ increase faster than $(\rho_{1},\rho_{2},\rho_{12})=(0.9,0.9,0.1)$. For the increasing-trend alternative, overall powers for the intervention effect increase at faster rate than the time effect. Also, powers of the time effect increase faster for $(\rho_{1},\rho_{2},\rho_{12})=(0.1,0.1,0.9)$ or $(0.1,0.9,0.9)$ compared to $(\rho_{1},\rho_{2},\rho_{12})=(0.9,0.9,0.1)$. {In summary, large inter-cluster correlations together with small intra-cluster correlations produce the smallest powers while small inter-cluster correlations together with large intra-cluster correlations produce the largest powers for tests on time effects.} All the aforementioned patterns in the powers  carry on for  data  generated from the multivariate log-normal distribution. Powers for multivariate Cauchy data generally maintain the same patterns as previously discussed, while having smaller rate of increase compared to the other distributions. This is phenomenon are clearly noticeable for all three alternatives.

\begin{table}
	\centering
	\caption{Achieved power ($\times100$) for the ANOVA-type tests on intervention, time and interaction effects based on the $\chi^2$ approximation for the one-point alternative and balanced design ($n_c=n_1=n_2=5$). Data are generated from discretized multivariate normal distribution with homogeneous variances ($\sigma_1^2=\sigma_2^2=1$). Here, $\delta$ is the shift coefficient.}
	\begin{tabular}{|c|c|c|c|c|c|c|}
		\hline
		\multicolumn{7}{|c|}{$(\rho_{1},\rho_{2},\rho_{12})=(0.9,0.9,0.1)$}\\ \hline
		&\multicolumn{3}{c|}{$M=3$} & \multicolumn{3}{c|}{$M=6$} \\ \hline
		$\delta$ & Intervention & Time & Interaction & Intervention & Time & Interaction \\ \hline
		0 & 5.8  & 5.1  & 6.0    & 6.8  & 5.2  & 6.0 \\ \hline
		0.3 & 8.1  & 6.8  & 8.5  & 8.1  & 6.6  & 9.4 \\ \hline
		0.6 & 12.6 & 11.4 & 13.3 & 14.1 & 11.6 & 15.3 \\ \hline
		0.9 & 23.4 & 18.8 & 23.5 & 24.2 & 20.8 & 25.4 \\ \hline
		1.2 & 36.6 & 29.1 & 40.2 & 36.5 & 29.5 & 39.5 \\ \hline
		1.5 & 51.9 & 39.4 & 56.3 & 51.8 & 41.1 & 56.4 \\ \hline
		1.8 & 66.0   & 51.0   & 70.4 & 66.8 & 52.5 & 71.8 \\ \hline
		2.1 & 77.9 & 60.6 & 81.4 & 79.2 & 61.6 & 82.0 \\ \hline
		2.4 & 85.8 & 68.8 & 88.7 & 86.6 & 69.5 & 87.9 \\ \hline
		2.7 & 91.8 & 73.5 & 92.4 & 91.2 & 74.8 & 93.4 \\ \hline
		3.0 & 95.0   & 78.4 & 95.5 & 94.3 & 79.5 & 96.1 \\ \hline
		\multicolumn{7}{|c|}{$(\rho_{1},\rho_{2},\rho_{12})=(0.1,0.1,0.9)$} \\ \hline
		& \multicolumn{3}{c|}{$M=3$} & \multicolumn{3}{c|}{$M=6$} \\ \hline
		$\delta$ & Intervention & Time & Interaction & Intervention & Time & Interaction \\ \hline
		0 & 5.6  & 4.2  & 3.9  & 5.8  & 4.9  & 4.0 \\ \hline
		0.3 & 7.4  & 6.8  & 7.0    & 8.3  & 9.6  & 7.9 \\ \hline
		0.6 & 12.1 & 13.5 & 18.2 & 15.2 & 19.8 & 19.6 \\ \hline
		0.9 & 21.9 & 27.0   & 38.3 & 26.3 & 33.7 & 40.7 \\ \hline
		1.2 & 34.5 & 44.7 & 60.5 & 42.7 & 49.3 & 67.2 \\ \hline
		1.5 & 48.2 & 61.9 & 79.7 & 59.9 & 67.9 & 87.9 \\ \hline
		1.8 & 62.2 & 76.3 & 91.4 & 75.5 & 80.0   & 95.7 \\ \hline
		2.1 & 75.3 & 85.5 & 96.8 & 85.3 & 88.7 & 98.5 \\ \hline
		2.4 & 84.0   & 91.5 & 99.1 & 91.6 & 93.8 & 99.5 \\ \hline
		2.7 & 90.0   & 94.0   & 99.6 & 95.6 & 96.5 & 99.9 \\ \hline
		3.0 & 93.3 & 95.6 & 99.8 & 98.2 & 97.6 & 100 \\ \hline
		\multicolumn{7}{|c|}{$(\rho_{1},\rho_{2},\rho_{12})=(0.1,0.9,0.9)$} \\ \hline
		& \multicolumn{3}{c|}{$M=3$} & \multicolumn{3}{c|}{$M=6$} \\ \hline
		$\delta$ & Intervention & Time & Interaction & Intervention & Time & Interaction \\ \hline
		0 & 5.9  & 3.8  & 3.8  & 5.7  & 5.9  & 4.1 \\ \hline
		0.3 & 7.0    & 6.6  & 5.3  & 8.1  & 8.3  & 7.2 \\ \hline
		0.6 & 12.1 & 13.1 & 16.3 & 14.1 & 17.5 & 18.2 \\ \hline
		0.9 & 19.7 & 24.7 & 36.2 & 23.0   & 29.0   & 36.6 \\ \hline
		1.2 & 32.7 & 41.3 & 59.5 & 36.4 & 44.8 & 61.4 \\ \hline
		1.5 & 45.8 & 59.0   & 76.1 & 50.0   & 61.5 & 81.2 \\ \hline
		1.8 & 58.8 & 72.2 & 89.3 & 64.9 & 73.8 & 92.8 \\ \hline
		2.1 & 70.6 & 82.1 & 95.1 & 77.5 & 82.0   & 97.7 \\ \hline
		2.4 & 81.0   & 88.3 & 97.7 & 85.5 & 87.4 & 99.2 \\ \hline
		2.7 & 87.6 & 91.3 & 99.0   & 91.6 & 91.2 & 99.8 \\ \hline
		3.0 & 91.9 & 93.6 & 99.7 & 95.3 & 93.5 & 99.9 \\ \hline
	\end{tabular}
	\label{Table:Power-OnePoint-Homo-Normal}
\end{table}

\begin{table}
	\centering
	\caption{Achieved power ($\times100$) for the ANOVA-type tests on intervention, time and interaction effects based on the $\chi^2$ approximation for the one-time alternative and balanced design ($n_c=n_1=n_2=5$). Data are generated from discretized multivariate normal distribution with homogeneous variances ($\sigma_1^2=\sigma_2^2=1$). Here, $\delta$ is the shift coefficient.}
	\begin{tabular}{|c|c|c|c|c|c|c|}
		\hline
		\multicolumn{7}{|c|}{$(\rho_{1},\rho_{2},\rho_{12})=(0.9,0.9,0.1)$}\\ \hline
		&\multicolumn{3}{c|}{$M=3$} & \multicolumn{3}{c|}{$M=6$} \\ \hline
		$\delta$ & Intervention & Time & Interaction & Intervention & Time & Interaction \\ \hline
		0 & 5.5 & 5 & 5.6 & 5.2 & 4.9     & 6.5 \\ \hline
		0.3 & 5.5 & 19.5     & 6.2 & 6.1 & 20.3    & 6.7 \\ \hline
		0.6 & 5.2 & 57.8     & 5.6 & 5.6 & 58.7    & 6.8 \\ \hline
		0.9 & 4.9 & 90.6     & 5.9 & 5.2 & 88.5    & 6.8 \\ \hline
		1.2 & 5.0   & 99.0       & 6.1 & 5.7 & 99.3    & 6.5 \\ \hline
		1.5 & 4.6 & 99.8     & 5.3 & 4.7 & 99.9    & 6.6 \\ \hline
		1.8 & 4.4 & 100      & 5.2 & 5.0   & 100     & 5.7 \\ \hline
		2.1 & 5.2 & 100      & 5.5 & 5.5 & 100     & 6.0 \\ \hline
		2.4 & 4.8 & 100      & 6.0   & 5.0   & 100     & 6.2 \\ \hline
		2.7 & 4.4 & 99.8     & 5.1 & 5.3 & 100     & 5.9 \\ \hline
		3.0 & 4.9 & 99.3 & 5.3 & 5.0   & 99.8 & 5.9 \\ \hline
		\multicolumn{7}{|c|}{$(\rho_{1},\rho_{2},\rho_{12})=(0.1,0.1,0.9)$} \\ \hline
		& \multicolumn{3}{c|}{$M=3$} & \multicolumn{3}{c|}{$M=6$} \\ \hline
		$\delta$ & Intervention & Time & Interaction & Intervention & Time & Interaction \\ \hline
		0 & 5.2 & 3.3 & 4.0   & 5.9 & 5.0    & 4.1 \\ \hline
		0.3 & 5.6 & 28.2     & 3.5 & 5.3 & 33.4 & 3.7 \\ \hline
		0.6 & 5.4 & 82.5     & 3.2 &  6.2 & 84.6 & 3.8 \\ \hline
		0.9 & 4.9 & 99.0       & 4.2 & 5.8 & 99.5 & 3.7 \\ \hline
		1.2 & 4.9 & 100      & 3.6 & 5.6 & 100  & 3.9 \\ \hline
		1.5 & 5.6 & 100 & 3.7 & 5.3 & 100  & 3.6 \\ \hline
		1.8 & 4.0 & 100 & 3.2 & 5.8 & 100  & 4.0\\ \hline
		2.1 & 5.2 & 100 & 3.2 & 5.5 & 100  & 4.8 \\ \hline
		2.4 & 5.0 & 100 & 3.8 & 5.0   & 100  & 3.9 \\ \hline
		2.7 & 4.6 & 100 & 3.1 & 5.9 & 100  & 3.5 \\ \hline
		3.0 &  3.7 & 99.7 & 4.0 & 5.6 & 100  & 3.9 \\ \hline
		\multicolumn{7}{|c|}{$(\rho_{1},\rho_{2},\rho_{12})=(0.1,0.9,0.9)$} \\ \hline
		& \multicolumn{3}{c|}{$M=3$} & \multicolumn{3}{c|}{$M=6$} \\ \hline
		$\delta$ & Intervention & Time & Interaction & Intervention & Time & Interaction \\ \hline
		0 & 4.8 & 4.6  & 3.5 & 6.7 & 4.8  & 4.5 \\ \hline
		0.3 & 5.2 & 27.4 & 3.6 & 6.4 & 28.8 & 4.7 \\ \hline
		0.6 & 5.3 & 78.9 & 3.6 & 6.8 & 79.1 & 4.2 \\ \hline
		0.9 & 4.7 & 98.6 & 3.8 & 6.2 & 97.9 & 4.9 \\ \hline
		1.2 & 4.6 & 99.9 & 4.0   & 6.1 & 100  & 4.9 \\ \hline
		1.5 & 4.9 & 100  & 3.8 & 6.7 & 100 & 4.8\\ \hline
		1.8 & 5.0   & 100  & 3.7 & 6.4 & 100 & 4.1\\ \hline
		2.1 & 4.6 & 100  & 3.8 & 6.4 & 100 & 4.4\\ \hline
		2.4 & 5.0   & 100  & 3.9 & 6.9 & 100 & 3.6 \\ \hline
		2.7 & 5.1 & 100  & 3.9 & 6.6 & 100 & 4.2\\ \hline
		3.0 & 4.5 & 99.7 & 4.0 & 6.6 & 100 & 4.6 \\ \hline
	\end{tabular}
	\label{Table:Power-OneTime-Homo-Normal}
\end{table}

\begin{table}
	\centering
	\caption{Achieved power ($\times100$) for the ANOVA-type tests on intervention, time and interaction effects based on the $\chi^2$ approximation for the increasing-trend alternative and balanced design ($n_c=n_1=n_2=5$). Data are generated from discretized multivariate normal distribution with homogeneous variances ($\sigma_1^2=\sigma_2^2=1$). Here, $\delta$ is the shift coefficient.}
	\begin{tabular}{|c|c|c|c|c|c|c|}
		\hline
		\multicolumn{7}{|c|}{$(\rho_{1},\rho_{2},\rho_{12})=(0.9,0.9,0.1)$}\\ \hline
		&\multicolumn{3}{c|}{$M=3$} & \multicolumn{3}{c|}{$M=6$} \\ \hline
		$\delta$ & Intervention & Time & Interaction & Intervention & Time & Interaction \\ \hline
		0 & 6.6  & 4.7  & 6.3 & 5.4  & 5.3  & 5.8 \\ \hline
		0.3 & 9.5  & 5.2  & 5.6 & 10.0   & 5.7  & 6.8 \\ \hline
		0.6 & 17.7 & 5.5  & 6.4 & 17.0   & 6.6  & 5.7 \\ \hline
		0.9 & 31.2 & 8.0    & 6.2 & 29.9 & 8.9  & 5.3 \\ \hline
		1.2 & 49.5 & 9.6  & 5.8 & 50.9 & 12.6 & 5.2 \\ \hline
		1.5 & 69.6 & 13.5 & 6.1 & 72.5 & 15.9 & 6.4 \\ \hline
		1.8 & 85.2 & 18.0   & 5.4 & 87.7 & 19.8 & 5.3 \\ \hline
		2.1 & 94.2 & 23.5 & 6.4 & 95.1 & 25.0   & 4.9 \\ \hline
		2.4 & 98.3 & 29.2 & 4.8 & 98.5 & 30.3 & 5.0 \\ \hline
		2.7 & 99.9 & 35.7 & 6.3 & 99.5 & 36.3 & 5.9 \\ \hline
		3.0 & 100  & 43.3 & 5.9 & 100  & 43.1 & 5.3 \\ \hline
		\multicolumn{7}{|c|}{$(\rho_{1},\rho_{2},\rho_{12})=(0.1,0.1,0.9)$} \\ \hline
		& \multicolumn{3}{c|}{$M=3$} & \multicolumn{3}{c|}{$M=6$} \\ \hline
		$\delta$ & Intervention & Time & Interaction & Intervention & Time & Interaction \\ \hline
		0 & 5.8  & 4.3  & 4.2 & 5.8  & 3.7  & 4.7 \\ \hline
		0.3 & 8.5  & 5.1  & 4.0   & 9.1  & 4.2  & 5.6 \\ \hline
		0.6 & 16.1 & 7.2  & 4.5 & 19.6 & 6.0    & 5.3 \\ \hline
		0.9 & 28.9 & 10.0   & 4.0   & 36.9 & 8.8  & 5.1 \\ \hline
		1.2 & 46.2 & 15.0   & 4.4 & 58.5 & 16.5 & 5.4 \\ \hline
		1.5 & 65.9 & 20.9 & 4.0   & 80.1 & 23.6 & 4.4 \\ \hline
		1.8 & 82.8 & 27.6 & 5.4 & 91.2 & 30.5 & 4.6 \\ \hline
		2.1 & 92.1 & 36.5 & 4.9 & 96.3 & 38.7 & 4.6 \\ \hline
		2.4 & 97.1 & 44.4 & 4.3 & 99.0   & 50.0   & 5.7 \\ \hline
		2.7 & 99.0   & 54.1 & 5.0   & 99.8 & 59.8 & 5.0 \\ \hline
		3.0 & 99.9 & 63.7 & 4.5 & 100  & 70.0   & 4.6 \\ \hline
		\multicolumn{7}{|c|}{$(\rho_{1},\rho_{2},\rho_{12})=(0.1,0.9,0.9)$} \\ \hline
		& \multicolumn{3}{c|}{$M=3$} & \multicolumn{3}{c|}{$M=6$} \\ \hline
		$\delta$ & Intervention & Time & Interaction & Intervention & Time & Interaction \\ \hline
		0 & 6.0    & 3.5  & 3.1 & 7.2  & 4.5  & 4.1 \\ \hline
		0.3 & 8.2  & 5.3  & 3.8 & 8.3  & 4.8  & 5.2 \\ \hline
		0.6 & 15.5 & 6.6  & 3.6 & 16.3 & 6.8  & 5.1 \\ \hline
		0.9 & 28.2 & 9.4  & 3.9 & 27.9 & 9.6  & 3.8 \\ \hline
		1.2 & 43.4 & 13.7 & 3.5 & 47.9 & 14.9 & 4.9 \\ \hline
		1.5 & 62.9 & 19.9 & 4.4 & 65.8 & 20.2 & 4.8 \\ \hline
		1.8 & 78.8 & 24.9 & 4.5 & 83.4 & 28.0   & 5.0 \\ \hline
		2.1 & 89.5 & 34.3 & 4.3 & 92.0 & 36.4 & 4.5 \\ \hline
		2.4 & 95.2 & 40.4 & 4.1 & 97.5 & 43.5 & 4.4 \\ \hline
		2.7 & 98.1 & 52.1 & 3.7 & 99.2 & 53.0 & 4.3 \\ \hline
		3.0 & 99.7 & 60.0   & 4.5 & 100 & 62.0 & 5.1\\ \hline
	\end{tabular}
	\label{Table:Power-IncreasingTrend-Homo-Normal}
\end{table}

\begin{table}
	\centering
	\caption{Achieved power ($\times100$) for the ANOVA-type tests on intervention, time and interaction effects based on the $\chi^2$ approximation for all three alternatives and balanced design ($n_c=n_1=n_2=5$). Data are generated from  multivariate log-normal distribution with homogeneous variances ($\sigma_1^2=\sigma_2^2=1$) and correlation coefficients $(\rho_{1},\rho_{2},\rho_{12})=(0.9,0.9,0.1)$. Here, $\delta$ is the shift coefficient.}
	\begin{tabular}{|c|c|c|c|c|c|c|}
		\hline
		\multicolumn{7}{|c|}{One-Point Alternative}\\ \hline
		&\multicolumn{3}{c|}{$M=3$} & \multicolumn{3}{c|}{$M=6$} \\ \hline
		$\delta$ & Intervention & Time & Interaction & Intervention & Time & Interaction \\ \hline
		0   & 5.5  & 5.1  & 4.2 & 5.9 & 4.9 & 4.8 \\ \hline
		0.3 & 6.4 & 6.0 & 6.0 & 7.4 & 6.1 & 7.7 \\ \hline
		0.6 & 12.3 & 12.0   & 12.7 & 14.1 & 10.8 & 14.1 \\ \hline
		0.9 & 21.0   & 20.2 & 22.8 & 24.3 & 18.4 & 26.3 \\ \hline
		1.2 & 34.4 & 30.9 & 38.8 & 37.7 & 29.4 & 41.3 \\ \hline
		1.5 & 50.8 & 43.6 & 57.2 & 53.1 & 38.8 & 58.0 \\ \hline
		1.8 & 64.9 & 52.1 & 70.6 & 65.0   & 49.8 & 70.6 \\ \hline
		2.1 & 78.2 & 59.6 & 81.5 & 77.0 & 58.4 & 81.5\\ \hline
		2.4 & 85.6 & 67.2 & 88.0 & 83.5 & 64.9 & 89.3 \\ \hline
		2.7 & 90.3 & 71.8 & 93.1 & 88.3 & 69.8 & 93.5 \\ \hline
		3.0   & 93.9 & 75.7 & 95.4 & 91.8 & 72.5 & 95.7\\ \hline
		\multicolumn{7}{|c|}{One-Time Alternative} \\ \hline
		& \multicolumn{3}{c|}{$M=3$} & \multicolumn{3}{c|}{$M=6$} \\ \hline
		$\delta$ & Intervention & Time & Interaction & Intervention & Time & Interaction \\ \hline
		0   & 6.2  & 4.9  & 4.0 &  5.1  & 4.8  & 5.4 \\ \hline
		0.3 & 6.0  & 19.5 & 3.9 & 5.2  & 22.4 & 5.2 \\ \hline
		0.6 & 5.5  & 58.1 & 4.7 & 5.6  & 60.8 & 5.2 \\ \hline
		0.9 & 5.7  & 91.0   & 4.3 & 5.5  & 90.9 & 4.8 \\ \hline
		1.2 & 5.5  & 99.1 & 4.1 & 5.2  & 99.2 & 4.8 \\ \hline
		1.5 & 5.5  & 99.9 & 4.0 & 4.9  & 100  & 4.6 \\ \hline
		1.8 & 5.1  & 100  & 3.8 & 4.6  & 100  & 4.8 \\ \hline
		2.1 & 5.2  & 100  & 3.9 & 4.8  & 100  & 5.0  \\ \hline
		2.4 & 5.2  & 100  & 4.2 & 5.0 & 100  & 5.1 \\ \hline
		2.7 & 5.3  & 99.7 & 4.0 & 5.1  & 100  & 5.0 \\ \hline
		3.0  & 5.3  & 98.8 & 3.8 & 5.2  & 99.1 & 5.1 \\ \hline
		\multicolumn{7}{|c|}{Increasing-Trend Alternative} \\ \hline
		& \multicolumn{3}{c|}{$M=3$} & \multicolumn{3}{c|}{$M=6$} \\ \hline
		$\delta$ & Intervention & Time & Interaction & Intervention & Time & Interaction \\ \hline
		0   & 6.2  & 4.9  & 4.0 & 5.1  & 4.8  & 5.4  \\ \hline
		0.3 & 7.0    & 5.3  & 4.0  & 6.7  & 5.8  & 5.6 \\ \hline
		0.6 & 15.4 & 5.9  & 4.0 & 17.2 & 7.1  & 5.4 \\ \hline
		0.9 & 30.0   & 7.9  & 3.8 & 31.0   & 8.6  & 5.2 \\ \hline
		1.2 & 49.4 & 10.1 & 4.0 & 51.0   & 13.2 & 5.3 \\ \hline
		1.5 & 68.9 & 14.0   & 4.0 & 69.7 & 17.2 & 5.2 \\ \hline
		1.8 & 85.0   & 19.0   & 3.9 & 85.0  & 21.4 & 5.0 \\ \hline
		2.1 & 94.3 & 25.1 & 3.6 & 94.8 & 26.3 & 5.1 \\ \hline
		2.4 & 98.3 & 30.0   & 4.0  & 97.8 & 32.8 & 5.1  \\ \hline
		2.7 & 99.5 & 35.6 & 3.9 & 99.6 & 39.1 & 5.5 \\ \hline
		3.0  & 99.8 & 41.0   & 4.3 & 99.9 & 46.4 & 5.3\\\hline
	\end{tabular}
	\label{Table:Power-Exponential}
\end{table}

\begin{table}
	\centering
	\caption{Achieved power ($\times100$) for the ANOVA-type tests on intervention, time and interaction effects based on the $\chi^2$ approximation for all three alternatives and balanced design ($n_c=n_1=n_2=5$). Data are generated from  multivariate Cauchy distribution with homogeneous variances ($\sigma_1^2=\sigma_2^2=1$) and correlation coefficients $(\rho_{1},\rho_{2},\rho_{12})=(0.9,0.9,0.1)$. Here, $\delta$ is the shift coefficient.}
	\begin{tabular}{|c|c|c|c|c|c|c|}
		\hline
		\multicolumn{7}{|c|}{One-Point Alternative}\\ \hline
		&\multicolumn{3}{c|}{$M=3$} & \multicolumn{3}{c|}{$M=6$} \\ \hline
		$\delta$ & Intervention & Time & Interaction & Intervention & Time & Interaction \\ \hline
		0 &  5.6 & 5.0 & 5.8 & 5.0 & 6.6 & 5.8\\ \hline
		0.3 & 6.2 & 5.7 & 6.5 & 6.0 & 6.6 & 6.2 \\ \hline
		0.6 &  7.6 & 7.2 & 7.9 & 8.0 & 7.7 & 7.6 \\ \hline
		0.9 &  9.4 & 9.9 & 11.4 & 10.9 & 10.7 & 10.6 \\ \hline
		1.2 &  13.6 & 13.5 & 15.1 &14.0 & 13.1 & 14.9 \\ \hline
		1.5 &  18.2 & 16.3 & 19.7 & 18.1 & 16.6 & 19.9 \\ \hline
		1.8 & 22.5 & 19.6 & 25.3 & 22.4 & 18.5 & 25.5 \\ \hline
		2.1 & 28.8 & 22.8 & 30.1 & 27.7 & 21.4 & 30.3 \\ \hline
		2.4 &33.6 & 26.4 & 34.2 & 33.1 & 24.4 & 36.1 \\ \hline
		2.7 & 37.4 & 30.3 & 38.0 & 37.8 & 27.3 & 39.7 \\ \hline
		3.0 & 41.6 & 33.4 & 42.0 & 42.4 & 30.6 & 44.7\\ \hline
		\multicolumn{7}{|c|}{One-Time Alternative} \\ \hline
		& \multicolumn{3}{c|}{$M=3$} & \multicolumn{3}{c|}{$M=6$} \\ \hline
		$\delta$ & Intervention & Time & Interaction & Intervention & Time & Interaction \\ \hline
		0 & 6.5 & 5.0 & 5.2 & 5.2 & 6.0 & 5.1 \\ \hline
		0.3 & 6.1 & 10.8 & 5.1 & 5.3 & 10.0 & 5.1 \\ \hline
		0.6  & 6.1 & 22.2 & 5.1 & 4.8 & 22.4 & 5.4\\ \hline
		0.9  & 6.1 & 43.3 & 5.1 & 5.2 & 39.8 & 5.8\\ \hline
		1.2  & 5.9 & 61.1 & 5.2 & 5.5 & 59.7 & 5.3\\ \hline
		1.5  & 5.7 & 74.8 & 4.6 & 5.4 & 75.2 & 5.3 \\ \hline
		1.8  & 5.6 & 84.3 & 5.0 & 5.3 & 83.6 & 5.4 \\ \hline
		2.1  & 5.4 & 90.5 & 4.8 & 5.1 & 89.3 & 5.3 \\ \hline
		2.4 & 5.6 & 94.0 & 4.6 & 4.5 & 93.2 & 4.9 \\ \hline
		2.7 & 5.3 & 96.1 & 4.4 & 4.4 & 95.3 & 4.9 \\ \hline
		3.0 & 5.2 & 97.4 & 4.3 & 4.6 & 97.1 & 4.7 \\ \hline
		\multicolumn{7}{|c|}{Increasing-Trend Alternative} \\ \hline
		& \multicolumn{3}{c|}{$M=3$} & \multicolumn{3}{c|}{$M=6$} \\ \hline
		$\delta$ & Intervention & Time & Interaction & Intervention & Time & Interaction \\ \hline
		0 & 6.5 & 5.0  & 5.2 & 5.2  & 6.0   & 5.1 \\ \hline
		0.3 & 7.1 & 5.6  & 4.7 & 5.7  & 6.5 & 5.0 \\ \hline
		0.6 & 9.6 & 6.1  & 5.0   & 9.4  & 6.1 & 5.0 \\ \hline
		0.9 & 12.8 & 6.8  & 4.4 & 14.5 & 6.5 & 5.0 \\ \hline
		1.2 & 18.9 & 7.4  & 4.5 &  20.3 & 7.5 & 5.3 \\ \hline
		1.5 & 26.0 & 8.6  & 4.4 &  27.6 & 8.8 & 5.0 \\ \hline
		1.8 & 34.9 & 9.8  & 4.7 &  37.1 & 9.8 & 5.6 \\ \hline
		2.1 & 42.8 & 11.6 & 4.5 & 45.0 & 10.9 & 5.3 \\ \hline
		2.4 & 52.3 & 13.4 & 4.2 & 52.7 & 11.9 & 5.2 \\ \hline
		2.7 & 60.1 & 14.1 & 4.3 & 60.6 & 15.0 & 4.2 \\ \hline
		3.0 & 66.7 &15.4 & 4.7 &  66.2 & 16.8 & 4.1 \\ \hline
	\end{tabular}
	\label{Table:Power-Cauchy}
\end{table}

\section{Analysis of the ARTIS Data}
\label{sec:AnalysisAndComparison}

\subsection{Nonparametric Analysis Based on Relative Effects}
\label{Sec:NonparmAnalysis}
Before applying the nonparametric test procedures, we pre-processed the ARTIS data to fit the data structure displayed in Table \ref{Table:Structure}. 
{As mentioned in Section \ref{sec:MotivatingExample}, there are three hierarchies of clustering. To  fit the data structure as described in Table \ref{Table:Structure} and illustrate the concepts in the proposed methods clearly, we keep only the second and third layers of  clustering  by randomly selecting one of the children for houses that have multiple participating children.}  Each selected child is regarded as a cluster while the measurements from multiple visits constitute the dependent replicates (observations on subunits).  {One may also keep the first layer of the cluster correlation by keeping data from only one of the visits for all participating children in the house so that the houses serve as  clusters.} There are 6 households in the placebo intervention group that have multiple participating children. The wood-stove and air-filter intervention groups each has 4 households with multiple participating children.  Since there are only a small number of households with multiple children, the effect of randomly selecting a child is minimal. The objective of the analysis is to evaluate whether intervention assignment modified the pre-intervention to post-intervention change in PAQLQ scores.

The three domains have the same sample size allocations: $\bm{n_c}=(37,13,35)$, $\bm{n_1}=(3,2,7)$ and $\bm{n_2}=(0,1,0)$ on each domain variable, where $\bm{n_c}=(n_1^{(c)},n_2^{(c)},n_3^{(c)})$, $\bm{n_1}=(n_{11},n_{21},n_{31})$ and $\bm{n_2}=(n_{12},n_{22},n_{32})$. In terms of the total sample sizes, there are 40, 16 and 42 participants in placebo, wood-stove and air-filter groups respectively. Boxplots for the three domain variables for each intervention group and at each intervention period are shown in Figure \ref{Figure:Boxplots}. {Estimates and test reults based on the nonparametric effect size measures are displayed in Table \ref{Table:TestStat}.} At $\alpha=0.05$, we conclude that only the time effect is significant for each of the three domain variables, {i.e. the effects size estimates are significantly different before and after the interventions, but this difference does not vary significantly by interventions. In order to closely investigate the pre-post nonparametric effect size changes across the intervention groups, additional tests were conducted and the results are displayed in Table \ref{Table:NonparametricContrast}. 
The nonparametric effect sizes on the activity limitation domain significantly increased after upgrading the wood stoves. On the emotional function domain, the nonparametric effect sizes significantly increased both after installing the sham air-filter (placebo) and upgrading the wood stoves. For the the symptoms domain, there is significant increase in the nonparametric effect size measure after installing the sham air-filter. Looking at the boxplots in Figure \ref{Figure:Boxplots}, we see that the "Post" boxplots are generally positioned higher than the corresponding "Pre" boxplots, though the tendency of "increase" is not always significant for all the interventions. 
}


However, since none of the tests on intervention effect is significant, we conclude that the PAQLQ scores in any intervention group do not tend to be greater or smaller than those in the other two intervention groups. This conclusion is valid for all three domain variables. There is no significant interaction effect, which means the intervention assignment does not modify the pre-intervention to post-intervention change in PAQLQ scores. To better visualize these conclusions, the boxplots of PAQLQ scores for each intervention group is displayed in Figure \ref{Figure:BoxplotsNoTime}. In this figure, the three boxes in each plot have quite similar ranges and quartiles and , thus, they do not generally appear to be different from each other. 

\begin{figure}[!htb]
	\caption{Boplots of QoL Scores by Intervention Group and Intervention Period}
	\includegraphics[scale=0.3]{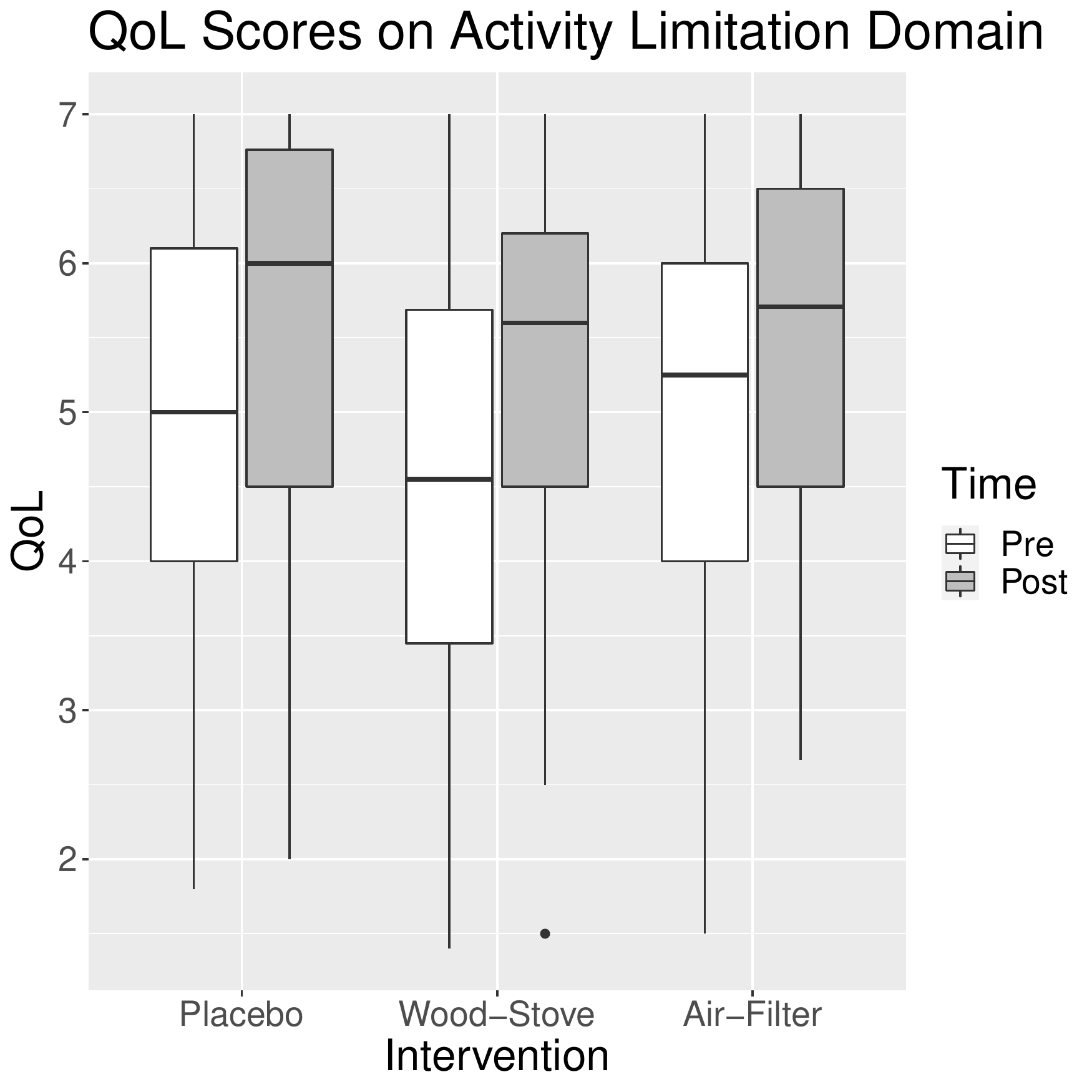}
	\includegraphics[scale=0.3]{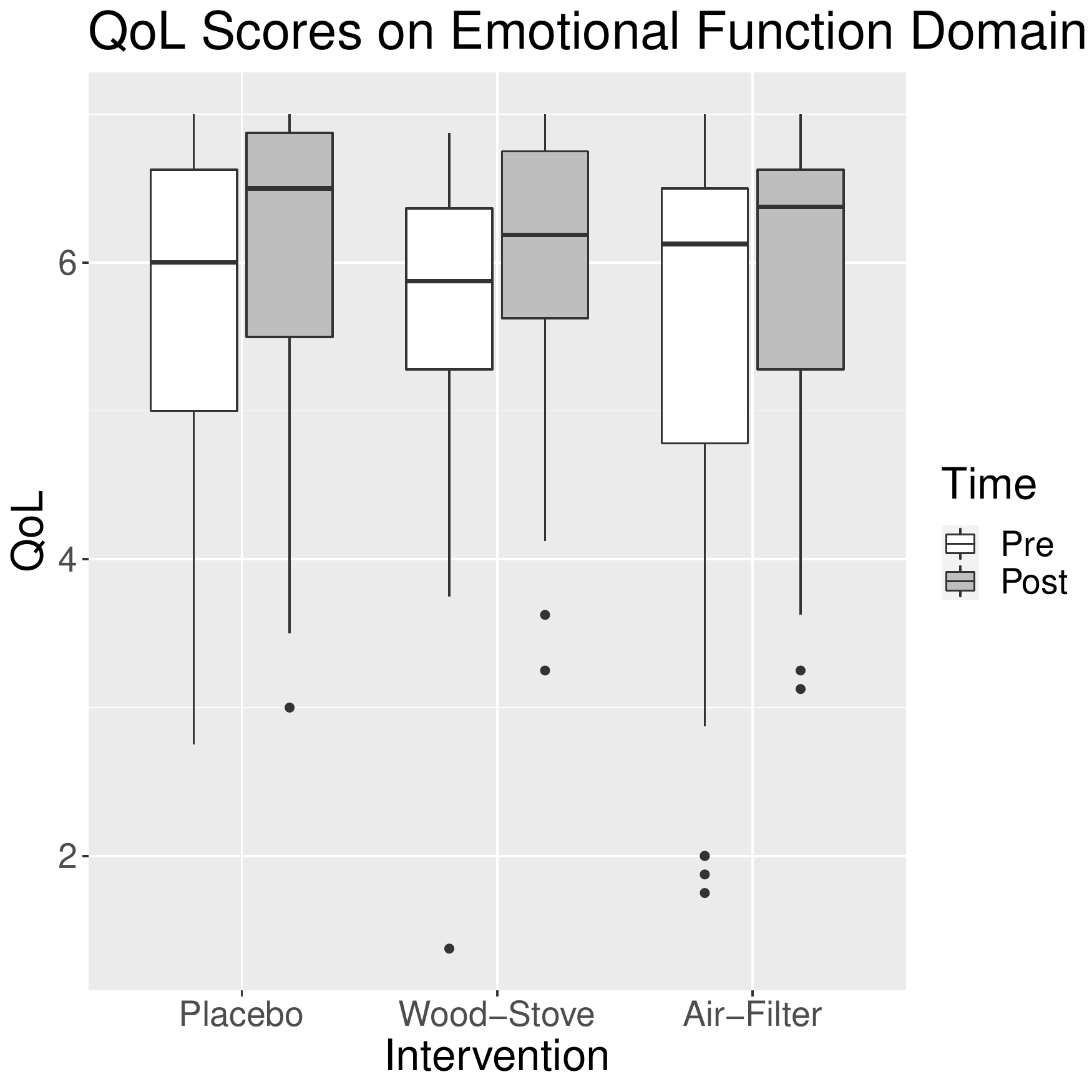}
	\includegraphics[scale=0.3]{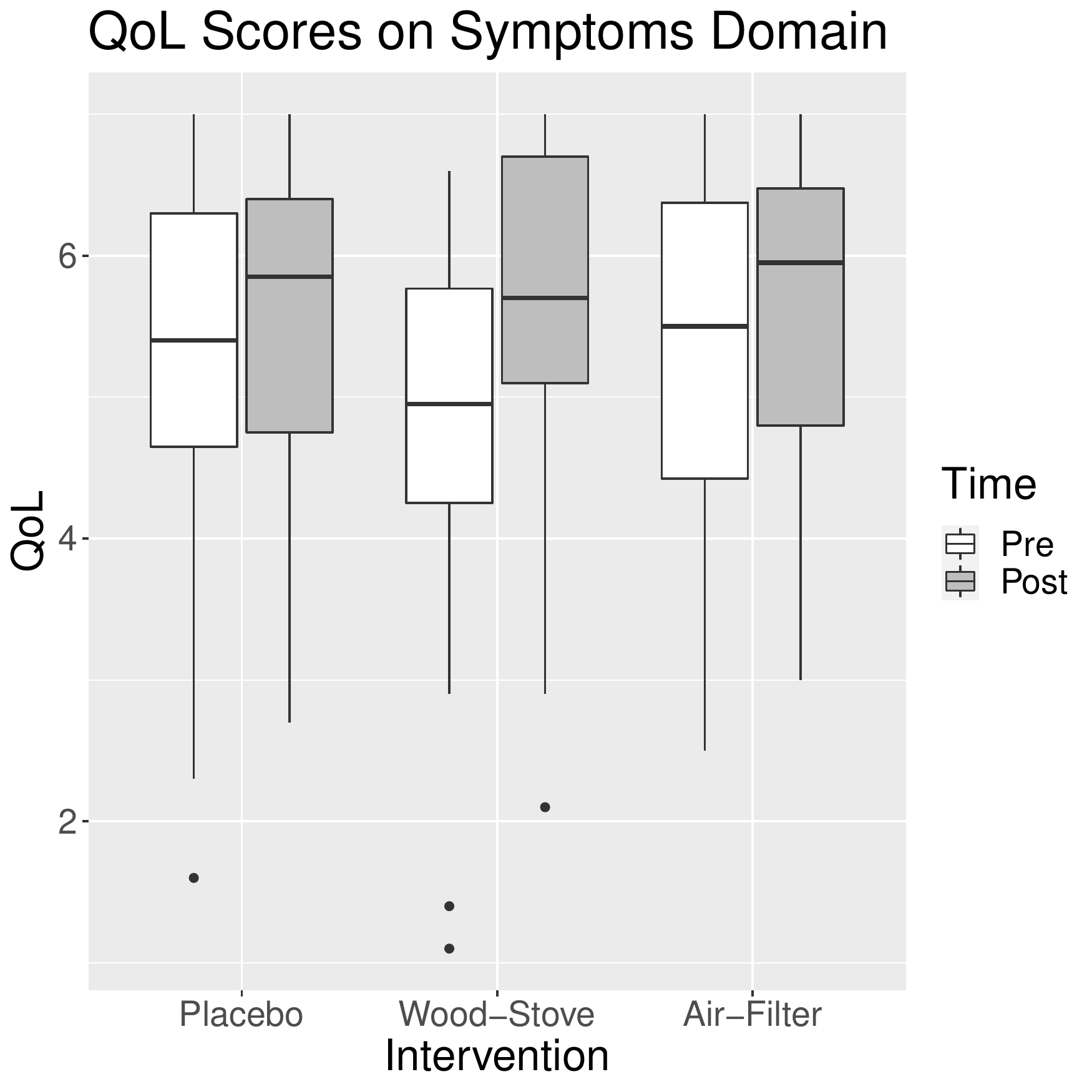}
	\label{Figure:Boxplots}
\end{figure}


\begin{table}[!htb]
	\caption{Nonparametric Effect Size Analysis for Domain Variables in ARTIS Data}
	\centering
	{
\begin{tabular}{|c|c|c|c|c|}
	\hline
	\multirow{2}{*}{\begin{tabular}[c]{@{}c@{}}Activity\\ Limitation\end{tabular}} & Effects                                                                  & Intervention        & Time            & Interaction        \\ \cline{2-5} 
	& $p$-value                                                               & 0.661            & 0.002                   & 0.237              \\ \hline
	\multirow{2}{*}{\begin{tabular}[c]{@{}c@{}}Emotional\\ Function\end{tabular}}  & Test                                                                  & Intervention        & Time            & Interaction        \\ \cline{2-5} 
	& $p$-value                                                               & 0.484            & \textless{}0.001        & 0.632              \\ \hline
	\multirow{3}{*}{Symptoms}                                                      & Effects                                                                & Intervention        & Time            & Interaction        \\ \cline{2-5} 
	& $p$-value                                                               & 0.738            & 0.005                   & 0.851              \\ \hline
\end{tabular}	
}
\label{Table:TestStat}
\end{table}

\begin{table}[!htb]
	\centering
	\tiny
	\caption{Tests on Pre-Post Nonparametric Effect Sizes by Intervention Group}
	{
		\begin{tabular}{|c|c|c|c|c|c|c|c|c|c|c|c|c|}
			\hline
			\begin{tabular}[c]{@{}c@{}}Nonparametric\\ Effect Size\end{tabular} & \multicolumn{4}{c|}{Activity Limitation}                                                                                                                                               & \multicolumn{4}{c|}{Emotional Function}                                                                                                                                                         & \multicolumn{4}{c|}{Symptoms}                                                                                                                                                          \\ \hline
			\begin{tabular}[c]{@{}c@{}}Intervention\\ Group\end{tabular}        & Pre                                                     & Post                                                    & Diff                                                     & p-value & Pre                                                     & Post                                                    & Diff                                                     & p-value          & Pre                                                     & Post                                                    & Diff                                                     & p-value \\ \hline
			Placebo                                                             & \begin{tabular}[c]{@{}c@{}}0.484\\ (0.589)\end{tabular} & \begin{tabular}[c]{@{}c@{}}0.547\\ (0.614)\end{tabular} & \begin{tabular}[c]{@{}c@{}}-0.063\\ (0.801)\end{tabular} & 0.132   & \begin{tabular}[c]{@{}c@{}}0.472\\ (0.509)\end{tabular} & \begin{tabular}[c]{@{}c@{}}0.588\\ (0.647)\end{tabular} & \begin{tabular}[c]{@{}c@{}}-0.116\\ (0.674)\end{tabular} & 0.001 & \begin{tabular}[c]{@{}c@{}}0.472\\ (0.570)\end{tabular} & \begin{tabular}[c]{@{}c@{}}0.577\\ (0.666)\end{tabular} & \begin{tabular}[c]{@{}c@{}}-0.105\\ (0.817)\end{tabular} & 0.015   \\ \hline
			Wood-stove                                                          & \begin{tabular}[c]{@{}c@{}}0.374\\ (0.846)\end{tabular} & \begin{tabular}[c]{@{}c@{}}0.560\\ (1.763)\end{tabular} & \begin{tabular}[c]{@{}c@{}}-0.186\\ (1.492)\end{tabular} & 0.018  & \begin{tabular}[c]{@{}c@{}}0.385\\ (1.035)\end{tabular} & \begin{tabular}[c]{@{}c@{}}0.529\\ (1.095)\end{tabular} & \begin{tabular}[c]{@{}c@{}}-0.144\\ (1.255)\end{tabular} & 0.030            & \begin{tabular}[c]{@{}c@{}}0.434\\ (0.714)\end{tabular} & \begin{tabular}[c]{@{}c@{}}0.518\\ (1.377)\end{tabular} & \begin{tabular}[c]{@{}c@{}}-0.084\\ (1.282)\end{tabular} & 0.209   \\ \hline
			Air-filter                                                          & \begin{tabular}[c]{@{}c@{}}0.483\\ (0.559)\end{tabular} & \begin{tabular}[c]{@{}c@{}}0.552\\ (0.696)\end{tabular} & \begin{tabular}[c]{@{}c@{}}-0.069\\ (0.881)\end{tabular} & 0.141   & \begin{tabular}[c]{@{}c@{}}0.474\\ (0.488)\end{tabular} & \begin{tabular}[c]{@{}c@{}}0.552\\ (0.671)\end{tabular} & \begin{tabular}[c]{@{}c@{}}-0.078\\ (0.876)\end{tabular} & 0.090            & \begin{tabular}[c]{@{}c@{}}0.467\\ (0.539)\end{tabular} & \begin{tabular}[c]{@{}c@{}}0.532\\ (0.635)\end{tabular} & \begin{tabular}[c]{@{}c@{}}-0.065\\ (0.811)\end{tabular} & 0.123   \\ \hline
		\end{tabular}
	}
	\label{Table:NonparametricContrast}
\end{table}

\begin{figure}[!htb]
	\caption{Boxplots of QoL Scores by Intervention Groups}
	\includegraphics[scale=0.3]{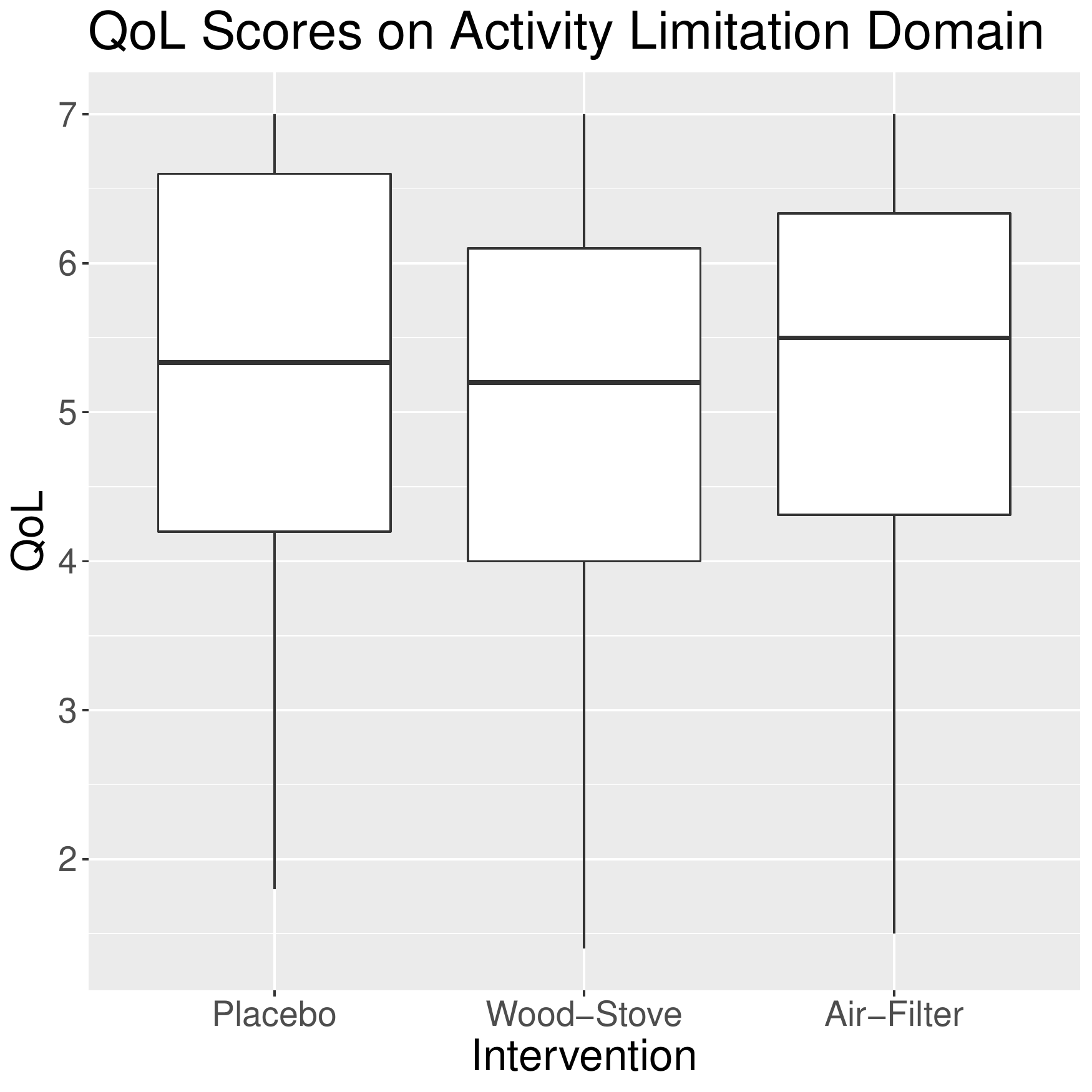}
	\includegraphics[scale=0.3]{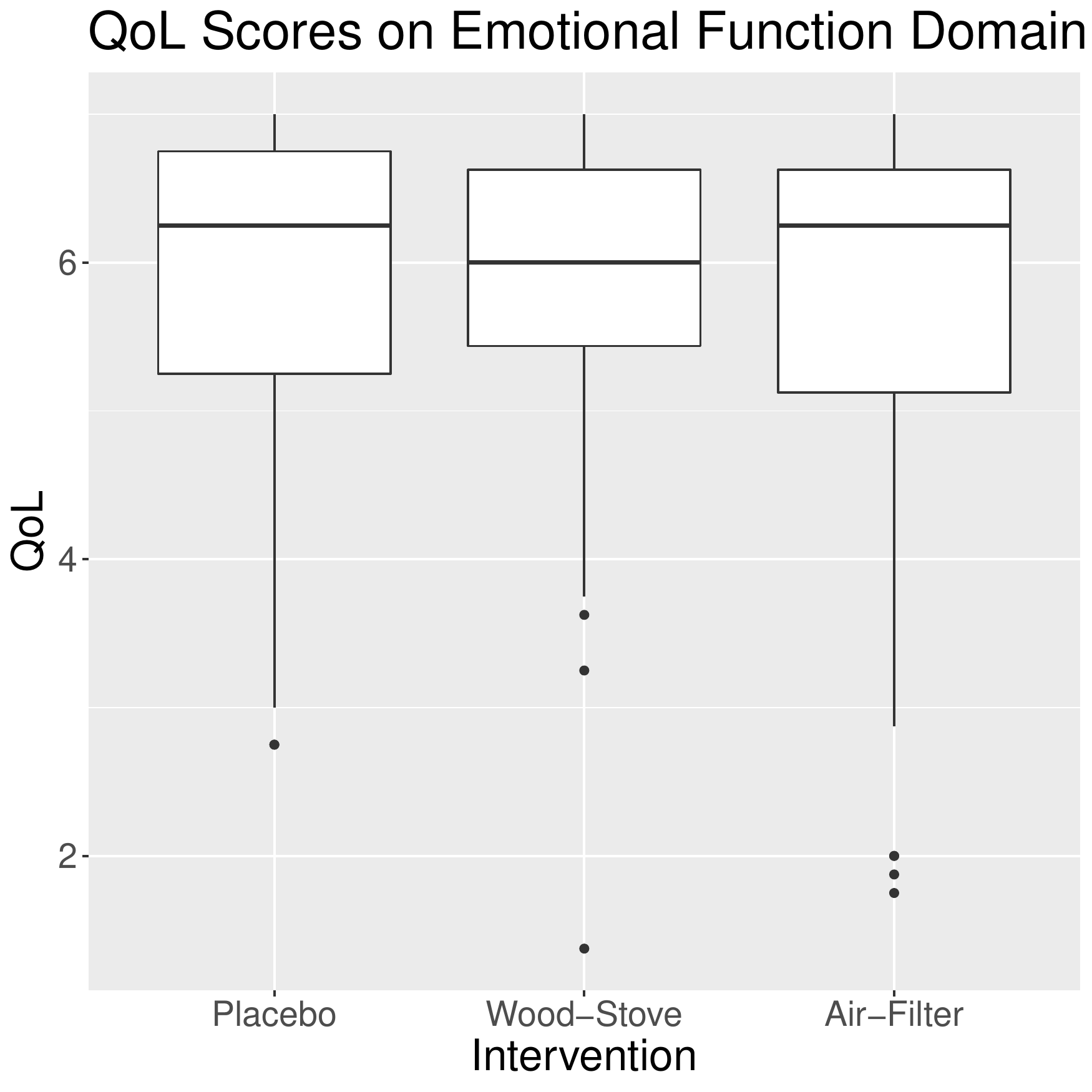}
	\includegraphics[scale=0.3]{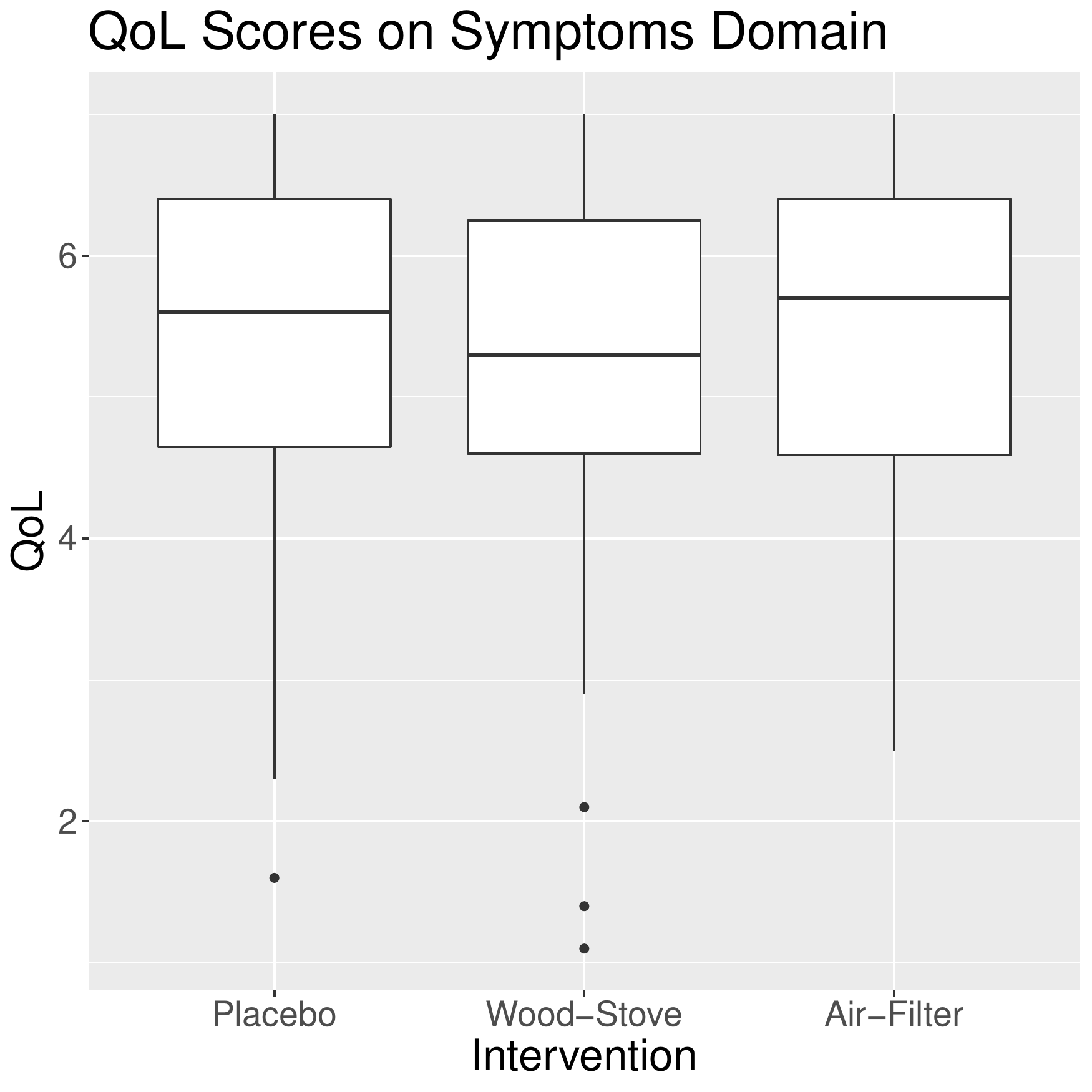}
	\label{Figure:BoxplotsNoTime}
\end{figure}

\subsection{Linear Mixed Effects Model}
\label{Sec:LinearMixedEffects}
For the purpose of comparison, we also analyze the ARTIS data with linear mixed effects model. 
As such, the linear mixed effects model includes fixed effects for intervention assignment and  time.  It also includes an interaction term  to assess effect modification by intervention assignment. Random effects for both the intercept and the slope (i.e. time) are included, which in effect allow each participant to have his/her own intercept and slope describing change in PAQLQ scores from the pre-intervention to post-intervention periods. Let $X_{jlkv}$ denote the observation for the $v^{th}$ visit, in the $k^{th}$ cluster (household), in the $j^{th}$ intervention group and during the $l^{th}$ intervention period.  The assumed mathematical model is 
\begin{equation*}
		X_{jlkv}=\mu+\delta_{j}+\gamma_{l}+(\delta\gamma)_{jl}+u_{l}+R_{jk}+\epsilon_{jlkv},
\end{equation*}
\begin{equation*}
	u_{l}\overset{iid}{\sim}N(0,\sigma_u^2),\quad R_{jk}\overset{iid}{\sim}N(0,\sigma_R^2)\quad \text{ and }\quad \epsilon_{jlkv}\overset{iid}{\sim}N(0,\sigma^2),
\end{equation*}
where $\delta_{j},\gamma_{l}$ and $(\delta\gamma)_{jl}$ are fixed effects for intervention, time and  interaction, respectively, $u_{l}$ and $R_{jk}$ are random effects for time and child, respectively, and $\epsilon_{jlkv}$ is the error term. The linear mixed effects model analysis amounts to utilizes likelihood-based available-case analysis for handling missing data. All analyses are conducted using \texttt{lme} function from \texttt{nlme} package. The  results are displayed in Tables \ref{Table:RandomMixedResults} and \ref{Table:ParametricContrast}. 

From Table \ref{Table:RandomMixedResults}, we observe that there are only significant time effects in all the three domains at level $\alpha=0.05$. {The pre-post detailed results in Table  \ref{Table:ParametricContrast} shows that the mean PAQLQ scores for the activity limitation domain significantly increased  after installing the sham air-filter and upgrading the wood stove. For the emotional function domain, the mean PAQLQ scores significantly increased only after installing the sham air-filter, while  the mean scores increased significantly only after upgrading the wood-stove for the symptoms domain.} 
	
{Though some of the results in Table \ref{Table:NonparametricContrast} and Table \ref{Table:ParametricContrast} are different, they both indicate there is some difference on PAQLQ scores before and after the interventions. The two models focus on different aspects and  they  provide useful complementary information. The nonparametric model gives us an overview of the tendency in responses variables by providing probabilistic statements on effects, whereas the linear mixed effects model gives specific comparisons on mean effects.}

\begin{table}[!htb]
	\caption{Linear Mixed Effects Model Analysis for Domain Variables in ARTIS data}
	\centering
	{
		\begin{tabular}{|c|c|c|c|c|}
			\hline
			\multirow{2}{*}{\begin{tabular}[c]{@{}c@{}}Activity\\ Limitation\end{tabular}} & Effects                                                                  & Intervention        & Time            & Interaction        \\ \cline{2-5} 
			& $p$-value                                                               & 0.666            & <0.001                   & 0.167              \\ \hline
			\multirow{2}{*}{\begin{tabular}[c]{@{}c@{}}Emotional\\ Function\end{tabular}}  & Test                                                                  & Intervention        & Time            & Interaction        \\ \cline{2-5} 
			& $p$-value                                                               & 0.883            & 0.001        & 0.901              \\ \hline
			\multirow{3}{*}{Symptoms}                                                      & Effects                                                                & Intervention        & Time            & Interaction        \\ \cline{2-5} 
			& $p$-value                                                               & 0.764            & 0.003                   & 0.217              \\ \hline
		\end{tabular}	
	}
	\label{Table:RandomMixedResults}
\end{table}

\begin{table}[!htb]
	\centering
	\tiny
	\caption{Tests on Pre-Post PAQLQ Scores by Intervention Group}
	{
	\begin{tabular}{|c|c|c|c|c|c|c|c|c|c|c|c|c|}
		\hline
		PAQLQ Scores       & \multicolumn{4}{c|}{Activity Limitation}                                                                                                                                             & \multicolumn{4}{c|}{Emotional Function}                                                                                                                                              & \multicolumn{4}{c|}{Symptoms}                                                                                                                                                        \\ \hline
		\begin{tabular}[c]{@{}c@{}}Intervention\\ Group\end{tabular} & Pre                                                    & Post                                                   & Diff                                                     & p-value & Pre                                                    & Post                                                   & Diff                                                     & p-value & Pre                                                    & Post                                                   & Diff                                                     & p-value \\ \hline
		Placebo            & \begin{tabular}[c]{@{}c@{}}4.94\\ (0.198)\end{tabular} & \begin{tabular}[c]{@{}c@{}}5.48\\ (0.213)\end{tabular} & \begin{tabular}[c]{@{}c@{}}-0.542\\ (0.182)\end{tabular} & 0.003   & \begin{tabular}[c]{@{}c@{}}5.66\\ (0.174)\end{tabular} & \begin{tabular}[c]{@{}c@{}}6.03\\ (0.160)\end{tabular} & \begin{tabular}[c]{@{}c@{}}-0.367\\ (0.149)\end{tabular} & 0.015   & \begin{tabular}[c]{@{}c@{}}5.27\\ (0.177)\end{tabular} & \begin{tabular}[c]{@{}c@{}}5.53\\ (0.177)\end{tabular} & \begin{tabular}[c]{@{}c@{}}-0.267\\ (0.173)\end{tabular} & 0.123   \\ \hline
		Wood-stove         & \begin{tabular}[c]{@{}c@{}}4.48\\ (0.319)\end{tabular} & \begin{tabular}[c]{@{}c@{}}5.39\\ (0.338)\end{tabular} & \begin{tabular}[c]{@{}c@{}}-0.907\\ (0.298)\end{tabular} & 0.003   & \begin{tabular}[c]{@{}c@{}}5.67\\ (0.279)\end{tabular} & \begin{tabular}[c]{@{}c@{}}6.01\\ (0.255)\end{tabular} & \begin{tabular}[c]{@{}c@{}}-0.333\\ (0.245)\end{tabular} & 0.175   & \begin{tabular}[c]{@{}c@{}}4.79\\ (0.285)\end{tabular} & \begin{tabular}[c]{@{}c@{}}5.63\\ (0.283)\end{tabular} & \begin{tabular}[c]{@{}c@{}}-0.838\\ (0.283)\end{tabular} & 0.003   \\ \hline
		Air-filter         & \begin{tabular}[c]{@{}c@{}}5.03\\ (0.194)\end{tabular} & \begin{tabular}[c]{@{}c@{}}5.29\\ (0.214)\end{tabular} & \begin{tabular}[c]{@{}c@{}}-0.257\\ (0.185)\end{tabular} & 0.166   & \begin{tabular}[c]{@{}c@{}}5.62\\ (0.170)\end{tabular} & \begin{tabular}[c]{@{}c@{}}5.89\\ (0.161)\end{tabular} & \begin{tabular}[c]{@{}c@{}}-0.271\\ (0.151)\end{tabular} & 0.074   & \begin{tabular}[c]{@{}c@{}}5.30\\ (0.173)\end{tabular} & \begin{tabular}[c]{@{}c@{}}5.52\\ (0.179)\end{tabular} & \begin{tabular}[c]{@{}c@{}}-0.217\\ (0.175)\end{tabular} & 0.217   \\ \hline
	\end{tabular}
}
\label{Table:ParametricContrast}
\end{table}

\section{Discussion and Conclusion}
\label{sec:DiscussionConclusion}
Nonparametric methods for the analysis of factorial designs is a topic of permanent interest in statistical research and applications, especially in cases where distributional assumptions of parametric models are not met. However, many of these nonparametric tests only apply when the hypotheses are formulated in terms of the distribution functions, i.e. $H_0^F(\textbf{T}):\textbf{TF}=\bm{0}$, where $\textbf{F}$ and $\textbf{T}=\textbf{C}^\top(\textbf{CC}^\top)^+\textbf{C}$ are the vector of distribution functions and unique projection pn the column space of the contrast matrix $\textbf{C}$. These hypotheses have some advantages. The resulting covariance matrix of the test statistic has quite simple structure with all of its components consistently estimated by functions of some versions of ranks, see \cite{akritas1997nonparametric} and \cite{akritas1997unified} for details. However, these hypotheses are quite restrictive in the sense that designs with heteroscedasticity are not covered and, the tests are not consistent to detect arbitrary alternatives $H_1^F(\textbf{T}):\textbf{TF}\ne\bm{0}$. Of course, lack of interpretability is another factor that contributes to the limitation of nonparametric methods based on distribution functions.

The aforementioned limitations of nonparametric methods are overcome by considering the more general hypotheses $H_0^p:\textbf{Tp}=\bm{0}$ \citep{brunner2017rank}. The work in \cite{brunner2017rank} is a generalization of the Behrens-Fisher problem to nonparametric factorial designs, which provides easily interpretable nonparametric treatment effect size measures. Though the general factorial designs with repeated measures are included in \cite{brunner2017rank}, their methods do not allow clustered data at each time point, not even to mention the more complicated data structures such as partially complete clustered data \citep{cui2020} which could possibly arise due to missing values. In this paper, we generalized \cite{cui2020} from a one-sample clustered data problem to multiple interventions situation. More specifically, nonparametric methods for clustered data in factorial design with pre-post intervention measurements are developed. The proposed nonparametric methods are applicable to designs with heteroscedasticity in the clustered data setting. Also, since elements of $\textbf{p}$ are fixed model quantities, our approach provides numerical and interpretable estimates of effect size measures along with their confidence intervals for further inference. The simulation studies show that the proposed methods preserve the pre-assigned Type-I error rates well and also achieve high powers under fairly reasonable alternatives. 

Though we are able to make intervention comparisons in the factorial layout, our methods are confined to a pre-post intervention design. A generalization of this nonparametric method to more than two intervention periods (time points) is not obvious and entails great difficulties. The main challenge comes from the rather involved structure of the covariance matrix of the effect size estimators. Actually this is a common issue for nonparametric tests that are constructed upon relative effect sizes. Also, as indicated in \cite{cui2020}, a more elaborate weighting scheme in the effect size estimators which includes the intra-cluster dependence can be more effective than the current weighting scheme. We plan to investigate these problems in future researches.

\bibliographystyle{chicago}
\bibliography{reference_new}

\begin{thebibliography}{}

\bibitem[\protect\citeauthoryear{Akritas and Arnold}{Akritas and
  Arnold}{1994}]{akritas1994fully}
Akritas, M.~G. and S.~F. Arnold (1994).
\newblock Fully nonparametric hypotheses for factorial designs i: Multivariate
  repeated measures designs.
\newblock {\em Journal of the American Statistical Association\/}~{\em
  89\/}(425), 336--343.

\bibitem[\protect\citeauthoryear{Akritas, Arnold, and Brunner}{Akritas
  et~al.}{1997}]{akritas1997nonparametric}
Akritas, M.~G., S.~F. Arnold, and E.~Brunner (1997).
\newblock Nonparametric hypotheses and rank statistics for unbalanced factorial
  designs.
\newblock {\em Journal of the American Statistical Association\/}~{\em
  92\/}(437), 258--265.

\bibitem[\protect\citeauthoryear{Akritas and Brunner}{Akritas and
  Brunner}{1997}]{akritas1997unified}
Akritas, M.~G. and E.~Brunner (1997).
\newblock A unified approach to rank tests for mixed models.
\newblock {\em Journal of Statistical Planning and Inference\/}~{\em 61\/}(2),
  249--277.

\bibitem[\protect\citeauthoryear{Boos and Brownie}{Boos and
  Brownie}{1992}]{boos1992rank}
Boos, D.~D. and C.~Brownie (1992).
\newblock A rank-based mixed model approach to multisite clinical trials.
\newblock {\em Biometrics\/}, 61--72.

\bibitem[\protect\citeauthoryear{Box et~al.}{Box et~al.}{1954}]{Box-1954}
Box, G.~E. et~al. (1954).
\newblock Some theorems on quadratic forms applied in the study of analysis of
  variance problems, i. effect of inequality of variance in the one-way
  classification.
\newblock {\em The annals of mathematical statistics\/}~{\em 25\/}(2),
  290--302.

\bibitem[\protect\citeauthoryear{Brunner, Dette, and Munk}{Brunner
  et~al.}{1997}]{ANOVA-1997}
Brunner, E., H.~Dette, and A.~Munk (1997).
\newblock Box-type approximations in nonparametric factorial designs.
\newblock {\em Journal of the American Statistical Association\/}~{\em
  92\/}(440), 1494--1502.

\bibitem[\protect\citeauthoryear{Brunner, Konietschke, Pauly, and Puri}{Brunner
  et~al.}{2017}]{brunner2017rank}
Brunner, E., F.~Konietschke, M.~Pauly, and M.~L. Puri (2017).
\newblock Rank-based procedures in factorial designs: Hypotheses about
  non-parametric treatment effects.
\newblock {\em Journal of the Royal Statistical Society: Series B (Statistical
  Methodology)\/}~{\em 79\/}(5), 1463--1485.

\bibitem[\protect\citeauthoryear{Brunner and Munzel}{Brunner and
  Munzel}{2000}]{baby-case-2000}
Brunner, E. and U.~Munzel (2000).
\newblock The nonparametric behrens-fisher problem: asymptotic theory and a
  small-sample approximation.
\newblock {\em Biometrical Journal: Journal of Mathematical Methods in
  Biosciences\/}~{\em 42\/}(1), 17--25.

\bibitem[\protect\citeauthoryear{Brunner and Puri}{Brunner and
  Puri}{2001}]{brunner2001nonparametric}
Brunner, E. and M.~L. Puri (2001).
\newblock Nonparametric methods in factorial designs.
\newblock {\em Statistical papers\/}~{\em 42\/}(1), 1--52.

\bibitem[\protect\citeauthoryear{Cui, Konietschke, and Harrar}{Cui
  et~al.}{2020}]{cui2020}
Cui, Y., F.~Konietschke, and S.~W. Harrar (2020).
\newblock The nonparametric behrens--fisher problem in partially complete
  clustered data.
\newblock {\em Biometrical Journal\/}~{\em 63\/}(1), 148--167.

\bibitem[\protect\citeauthoryear{Datta and Satten}{Datta and
  Satten}{2005}]{DS-2005}
Datta, S. and G.~A. Satten (2005).
\newblock Rank-sum tests for clustered data.
\newblock {\em Journal of the American Statistical Association\/}~{\em
  100\/}(471), 908--915.

\bibitem[\protect\citeauthoryear{De~Neve and Thas}{De~Neve and
  Thas}{2015}]{de2015regression}
De~Neve, J. and O.~Thas (2015).
\newblock A regression framework for rank tests based on the probabilistic
  index model.
\newblock {\em Journal of the American Statistical Association\/}~{\em
  110\/}(511), 1276--1283.

\bibitem[\protect\citeauthoryear{Hoffman, Sen, and Weinberg}{Hoffman
  et~al.}{2001}]{Hoffman2001}
Hoffman, E.~B., P.~K. Sen, and C.~R. Weinberg (2001).
\newblock Within-cluster resampling.
\newblock {\em Biometrika\/}~{\em 88\/}(4), 1121--1134.

\bibitem[\protect\citeauthoryear{Konietschke, Harrar, Lange, and
  Brunner}{Konietschke et~al.}{2012}]{matched-pair-2012}
Konietschke, F., S.~W. Harrar, K.~Lange, and E.~Brunner (2012).
\newblock Ranking procedures for matched pairs with missing data—asymptotic
  theory and a small sample approximation.
\newblock {\em Computational Statistics \& Data Analysis\/}~{\em 56\/}(5),
  1090--1102.

\bibitem[\protect\citeauthoryear{Kruskal et~al.}{Kruskal
  et~al.}{1952}]{Kruskal-1952}
Kruskal, W.~H. et~al. (1952).
\newblock A nonparametric test for the several sample problem.
\newblock {\em The Annals of Mathematical Statistics\/}~{\em 23\/}(4),
  525--540.

\bibitem[\protect\citeauthoryear{L\'evy}{L\'evy}{1925}]{Levy-1925}
L\'evy, P. (1925).
\newblock {\em Calcul des Probablit\'es}.
\newblock Pairs: Gauthier-Villars.

\bibitem[\protect\citeauthoryear{Mann and Whitney}{Mann and
  Whitney}{1947}]{WMW-1947}
Mann, H.~B. and D.~R. Whitney (1947).
\newblock On a test of whether one of two random variables is stochastically
  larger than the other.
\newblock {\em The annals of mathematical statistics\/}, 50--60.

\bibitem[\protect\citeauthoryear{Noonan, Semmens, Smith, Harrar, Montrose,
  Weiler, McNamara, and Ward}{Noonan et~al.}{2017}]{ARTIS-data}
Noonan, C.~W., E.~O. Semmens, P.~Smith, S.~W. Harrar, L.~Montrose, E.~Weiler,
  M.~McNamara, and T.~J. Ward (2017).
\newblock Randomized trial of interventions to improve childhood asthma in
  homes with wood-burning stoves.
\newblock {\em Environmental health perspectives\/}~{\em 125\/}(9), 097010.

\bibitem[\protect\citeauthoryear{Noonan and Ward}{Noonan and
  Ward}{2012}]{noonan2012asthma}
Noonan, C.~W. and T.~J. Ward (2012).
\newblock Asthma randomized trial of indoor wood smoke (artis): rationale and
  methods.
\newblock {\em Contemporary clinical trials\/}~{\em 33\/}(5), 1080--1087.

\bibitem[\protect\citeauthoryear{Ruymgaart}{Ruymgaart}{1980}]{Ruymgaart-1980}
Ruymgaart, F.~H. (1980).
\newblock A unified approach to the asymptotic distribution theory of certain
  midrank statistics.
\newblock In {\em Statistique non Parametrique Asymptotique}, pp.\  1--18.
  Springer.

\bibitem[\protect\citeauthoryear{Umlauft, Konietschke, and Pauly}{Umlauft
  et~al.}{2017}]{umlauft2017rank}
Umlauft, M., F.~Konietschke, and M.~Pauly (2017).
\newblock Rank-based permutation approaches for non-parametric factorial
  designs.
\newblock {\em British Journal of Mathematical and Statistical
  Psychology\/}~{\em 70\/}(3), 368--390.

\bibitem[\protect\citeauthoryear{Vallejo, Fern{\'a}ndez, and
  Livacic-Rojas}{Vallejo et~al.}{2010}]{vallejo2010analysis}
Vallejo, G., M.~Fern{\'a}ndez, and P.~E. Livacic-Rojas (2010).
\newblock Analysis of unbalanced factorial designs with heteroscedastic data.
\newblock {\em Journal of Statistical Computation and Simulation\/}~{\em
  80\/}(1), 75--88.

\bibitem[\protect\citeauthoryear{Ward, Semmens, Weiler, Harrar, and
  Noonan}{Ward et~al.}{2017}]{ward2017efficacy}
Ward, T.~J., E.~O. Semmens, E.~Weiler, S.~Harrar, and C.~W. Noonan (2017).
\newblock Efficacy of interventions targeting household air pollution from
  residential wood stoves.
\newblock {\em Journal of exposure science \& environmental
  epidemiology\/}~{\em 27\/}(1), 64--71.

\end{thebibliography}

\clearpage
\section{Appendix}
\label{sec:Appendix}
\subsection{Covariance Decomposition}
\label{Appendix:CovarianceDecomposition}
{\footnotesize
\begin{align*}
	&\textrm{Cov}(Z_{pqrs},Z_{p'q'rs})=\textrm{Cov}(Z_{rspq},Z_{rsp'q'})\allowdisplaybreaks[1]\\
	=&\textrm{Cov}\bigg(\frac{1}{N_{pq}}\sum_{k=1}^{n_{p}^{(c)}}\sum_{v=1}^{m_{pqk}^{(c)}}F_{rs}(X_{pqkv}^{(c)})-\frac{1}{N_{rs}}\sum_{k=1}^{n_{r}^{(c)}}\sum_{v=1}^{m_{rsk}^{(c)}}F_{pq}(X_{rskv}^{(c)})+\frac{1}{N_{pq}}\sum_{k=1}^{n_{pq}}\sum_{v=1}^{m_{pqk}^{(i)}}F_{rs}(X_{pqkv}^{(i)})-\frac{1}{N_{rs}}\sum_{k=1}^{n_{rs}}\sum_{v=1}^{m_{rsk}^{(i)}}F_{pq}(X_{rskv}^{(i)}),\allowdisplaybreaks[1]\\
	&\quad\frac{1}{N_{p'q'}}\sum_{k=1}^{n_{p'}^{(c)}}\sum_{v=1}^{m_{p'q'k}^{(c)}}F_{rs}(X_{p'q'kv}^{(c)})-\frac{1}{N_{rs}}\sum_{k=1}^{n_{r}^{(c)}}\sum_{v=1}^{m_{rsk}^{(c)}}F_{p'q'}(X_{rskv}^{(c)})+\frac{1}{N_{p'q'}}\sum_{k=1}^{n_{p'q'}}\sum_{v=1}^{m_{p'q'k}^{(i)}}F_{rs}(X_{p'q'kv}^{(i)})-\frac{1}{N_{rs}}\sum_{k=1}^{n_{rs}}\sum_{v=1}^{m_{rsk}^{(i)}}F_{p'q'}(X_{rskv}^{(i)})\bigg)\allowdisplaybreaks[1]\\
	=&\frac{1}{N_{pq}N_{p'q'}}\textrm{Cov}\bigg(\sum_{k=1}^{n_{p}^{(c)}}\sum_{v=1}^{m_{pqk}^{(c)}}F_{rs}(X_{pqkv}^{(c)}),\sum_{k=1}^{n_{p'}^{(c)}}\sum_{v=1}^{m_{p'q'k}^{(c)}}F_{rs}(X_{p'q'kv}^{(c)})\bigg)-\frac{1}{N_{pq}N_{rs}}\textrm{Cov}\bigg(\sum_{k=1}^{n_{p}^{(c)}}\sum_{v=1}^{m_{pqk}^{(c)}}F_{rs}(X_{pqkv}^{(c)}),\sum_{k=1}^{n_{r}^{(c)}}\sum_{v=1}^{m_{rsk}^{(c)}}F_{p'q'}(X_{rskv}^{(c)})\bigg)\allowdisplaybreaks[1]\\
	&+\frac{1}{N_{pq}N_{p'q'}}\textrm{Cov}\bigg(\sum_{k=1}^{n_{p}^{(c)}}\sum_{v=1}^{m_{pqk}^{(c)}}F_{rs}(X_{pqkv}^{(c)}),\sum_{k=1}^{n_{p'q'}}\sum_{v=1}^{m_{p'q'k}^{(i)}}F_{rs}(X_{p'q'kv}^{(i)})\bigg)-\frac{1}{N_{pq}N_{rs}}\textrm{Cov}\bigg(\sum_{k=1}^{n_{p}^{(c)}}\sum_{v=1}^{m_{pqk}^{(c)}}F_{rs}(X_{pqkv}^{(c)}),\sum_{k=1}^{n_{rs}}\sum_{v=1}^{m_{rsk}^{(i)}}F_{p'q'}(X_{rskv}^{(i)})\bigg)\allowdisplaybreaks[1]\\
	&-\frac{1}{N_{rs}N_{p'q'}}\textrm{Cov}\bigg(\sum_{k=1}^{n_{r}^{(c)}}\sum_{v=1}^{m_{rsk}^{(c)}}F_{pq}(X_{rskv}^{(c)}),\sum_{k=1}^{n_{p'}^{(c)}}\sum_{v=1}^{m_{p'q'k}^{(c)}}F_{rs}(X_{p'q'kv}^{(c)})\bigg)+\frac{1}{N_{rs}N_{rs}}\textrm{Cov}\bigg(\sum_{k=1}^{n_{r}^{(c)}}\sum_{v=1}^{m_{rsk}^{(c)}}F_{pq}(X_{rskv}^{(c)}),\sum_{k=1}^{n_{r}^{(c)}}\sum_{v=1}^{m_{rsk}^{(c)}}F_{p'q'}(X_{rskv}^{(c)})\bigg)\allowdisplaybreaks[1]\\
	&-\frac{1}{N_{rs}N_{p'q'}}\textrm{Cov}\bigg(\sum_{k=1}^{n_{r}^{(c)}}\sum_{v=1}^{m_{rsk}^{(c)}}F_{pq}(X_{rskv}^{(c)}),\sum_{k=1}^{n_{p'q'}}\sum_{v=1}^{m_{p'q'k}^{(i)}}F_{rs}(X_{p'q'kv}^{(i)})\bigg)+\frac{1}{N_{rs}N_{rs}}\textrm{Cov}\bigg(\sum_{k=1}^{n_{r}^{(c)}}\sum_{v=1}^{m_{rsk}^{(c)}}F_{pq}(X_{rskv}^{(c)}),\sum_{k=1}^{n_{rs}}\sum_{v=1}^{m_{rsk}^{(i)}}F_{p'q'}(X_{rskv}^{(i)})\bigg)\allowdisplaybreaks[1]\\
	&+\frac{1}{N_{pq}N_{p'q'}}\textrm{Cov}\bigg(\sum_{k=1}^{n_{pq}}\sum_{v=1}^{m_{pqk}^{(i)}}F_{rs}(X_{pqkv}^{(i)}),\sum_{k=1}^{n_{p'}^{(c)}}\sum_{v=1}^{m_{p'q'k}^{(c)}}F_{rs}(X_{p'q'kv}^{(c)})\bigg)-\frac{1}{N_{pq}N_{rs}}\textrm{Cov}\bigg(\sum_{k=1}^{n_{pq}}\sum_{v=1}^{m_{pqk}^{(i)}}F_{rs}(X_{pqkv}^{(i)}),\sum_{k=1}^{n_{r}^{(c)}}\sum_{v=1}^{m_{rsk}^{(c)}}F_{pq'}(X_{rskv}^{(c)})\bigg)\allowdisplaybreaks[1]\\
	&+\frac{1}{N_{pq}N_{p'q'}}\textrm{Cov}\bigg(\sum_{k=1}^{n_{pq}}\sum_{v=1}^{m_{pqk}^{(i)}}F_{rs}(X_{pqkv}^{(i)}),\sum_{k=1}^{n_{p'q'}}\sum_{v=1}^{m_{p'q'k}^{(i)}}F_{rs}(X_{p'q'kv}^{(i)})\bigg)-\frac{1}{N_{pq}N_{rs}}\textrm{Cov}\bigg(\sum_{k=1}^{n_{pq}}\sum_{v=1}^{m_{pqk}^{(i)}}F_{rs}(X_{pqkv}^{(i)}),\sum_{k=1}^{n_{rs}}\sum_{v=1}^{m_{rsk}^{(i)}}F_{p'q'}(X_{rskv}^{(i)})\bigg)\allowdisplaybreaks[1]\\
	&-\frac{1}{N_{rs}N_{p'q'}}\textrm{Cov}\bigg(\sum_{k=1}^{n_{rs}}\sum_{v=1}^{m_{rsk}^{(i)}}F_{pq}(X_{rskv}^{(i)}),\sum_{k=1}^{n_{p'}^{(c)}}\sum_{v=1}^{m_{p'q'k}^{(c)}}F_{rs}(X_{p'q'kv}^{(c)})\bigg)+\frac{1}{N_{rs}N_{rs}}\textrm{Cov}\bigg(\sum_{k=1}^{n_{rs}}\sum_{v=1}^{m_{rsk}^{(i)}}F_{rs}(X_{pqkv}^{(i)}),\sum_{k=1}^{n_{r}^{(c)}}\sum_{v=1}^{m_{rsk}^{(c)}}F_{p'q'}(X_{rskv}^{(c)})\bigg)\allowdisplaybreaks[1]\\
	&-\frac{1}{N_{rs}N_{p'q'}}\textrm{Cov}\bigg(\sum_{k=1}^{n_{rs}}\sum_{v=1}^{m_{rsk}^{(i)}}F_{pq}(X_{rskv}^{(i)}),\sum_{k=1}^{n_{p'q'}}\sum_{v=1}^{m_{p'q'k}^{(i)}}F_{rs}(X_{p'q'kv}^{(i)})\bigg)+\frac{1}{N_{rs}N_{rs}}\textrm{Cov}\bigg(\sum_{k=1}^{n_{rs}}\sum_{v=1}^{m_{rsk}^{(i)}}F_{pq}(X_{rskv}^{(i)}),\sum_{k=1}^{n_{rs}}\sum_{v=1}^{m_{rsk}^{(i)}}F_{p'q'}(X_{rskv}^{(i)})\bigg)\allowdisplaybreaks[1]\\
	&:=\frac{1}{N^2}\bigg[C_{1}-C_{2}+C_{3}-C_{4}-C_{5}+C_{6}-C_{7}+C_{8}+C_{9}-C_{10}+C_{11}-C_{12}-C_{13}+C_{14}-C_{15}+C_{16}\bigg]\\
\end{align*}}

\end{document}